\pdfoutput=1
\documentclass[10pt,pra,aps,superscriptaddress,nofootinbib,twocolumn,longbibliography]{revtex4-1}
\usepackage{amsmath,bbm}
\usepackage{amsthm}
\usepackage{amsmath}
\usepackage{latexsym}
\usepackage{amssymb}
\usepackage{float}
\usepackage{graphicx}           
\usepackage{color}
\usepackage{xcolor}
\usepackage{mathpazo}
\usepackage{comment}
\usepackage{enumitem}
\usepackage{multirow}
\usepackage{setspace}
\usepackage{subcaption}
\usepackage[colorlinks=true,linkcolor=blue,citecolor=magenta,urlcolor=blue]{hyperref}
\usepackage{svg}

\newcommand{\be}{\begin{equation}}
\newcommand{\ee}{\end{equation}}
\newcommand{\bea}{\begin{eqnarray}}
\newcommand{\eea}{\end{eqnarray}}

\def\squareforqed{\hbox{\rlap{$\sqcap$}$\sqcup$}}
\def\qed{\ifmmode\squareforqed\else{\unskip\nobreak\hfil
\penalty50\hskip1em\null\nobreak\hfil\squareforqed
\parfillskip=0pt\finalhyphendemerits=0\endgraf}\fi}
\def\endenv{\ifmmode\;\else{\unskip\nobreak\hfil
\penalty50\hskip1em\null\nobreak\hfil\;
\parfillskip=0pt\finalhyphendemerits=0\endgraf}\fi}

\newcommand{\tr}{\text{Tr}}
\newcommand{\I}{\mathbbm{1}}

\newcommand{\ket}[1]{|#1\rangle}
\newcommand{\bra}[1]{\langle#1|}

\newcommand{\la}{\langle}
\newcommand{\ra}{\rangle}

\makeatletter
\newtheorem*{rep@theorem}{\rep@title}
\newcommand{\newreptheorem}[2]{%
\newenvironment{rep#1}[1]{%
 \def\rep@title{#2 \ref{##1}}%
 \begin{rep@theorem}}%
 {\end{rep@theorem}}}
\makeatother

\newtheorem{thm}{Theorem}
\newreptheorem{thm}{Theorem}

\newtheorem{definition}{Definition}

\begin{document}

\title{Single-shot distinguishability and antidistinguishability of quantum measurements}


\author{Satyaki Manna}
\author{Sneha Suresh}
\author{Manan Singh Kachhawaha}
\author{Debashis Saha}
\affiliation{School of Physics, Indian Institute of Science Education and Research Thiruvananthapuram, Kerala 695551, India}


\begin{abstract}
Among the surprising features of quantum measurements, the problem of distinguishing and antidistinguishing general quantum measurements is fundamentally appealing. Unlike classical systems, quantum theory offers entangled states and peculiar state update rule of the post-measurement state, which gives rise to four distinct scenarios for distinguishing (and antidistinguishing) quantum measurements - $(i)$ probing single systems and without access to the post-measurement states, $(ii)$ probing entangled systems and without access to the post-measurement states, $(iii)$ probing single systems with access to the post-measurement states, and $(iv)$ probing entangled systems with access to the post-measurement states. In these scenarios, we consider the probability of distinguishing (and antidistinguishing) quantum measurements sampled from a given set in the single-shot regime. For some scenarios, we provide the analytical expressions of distinguishability (and antidistinguishability) for qubit projective measurements. We show that the distinguishability of any pair of qubit projective measurements in scenario $(iii)$ is always greater than its values in scenario $(ii)$. Interestingly, we find certain pairs of qubit non-projective measurements for which the optimal distinguishability in scenario $(ii)$ is achieved using a non-maximally entangled state. It turns out that, for any set of measurements, distinguishability (and antidistinguishability) in scenario $(i)$ is always less than or equal to in any other scenario, while it reaches its highest possible value in scenario $(iv)$. We establish that the relations form a strict hierarchy, and there is no hierarchical relation between scenarios $(ii)$ and $(iii)$. In particular, we introduce different variants of the well-known `trine' qubit measurement to construct pairs (and triples) of qubit quantum measurements such that they are perfectly distinguishable (and antidistinguishable) in scenario $(ii)$ but not in scenario $(iii)$, and vice versa. Additionally, we present qubit measurements that are perfectly distinguishable (and antidistinguishable) in scenario $(iv)$ but not in any other scenario.
\end{abstract}

\maketitle


\section{Introduction} \label{SEC I} 

The ability to distinguish between different physical processes establishes a fundamental limit to our understanding and observation of the physical world. In classical theory, the notion of distinguishability is straightforward. However, in quantum theory, which exhibits numerous counter-intuitive and non-classical phenomena, the concept becomes far more nuanced. Distinguishability, in general, refers to identifying which specific process has occurred from a set of possible processes. A related but weaker notion, called antidistinguishability, concerns identifying which process has not occurred from a given set.

Over the years, many aspects of distinguishability and antidistinguishability of quantum states have been studied extensively \cite{Hellstrom,Caves02,PhysRevA.93.062112,PhysRevA.70.022302,PhysRevA.88.052313,BandhopadhyayaPRA2014,johnston2023tight}. The role of distinguishability and antidistinguishability of quantum states on foundational aspects of quantum theory, such as the interpretation of the reality of quantum states, has been investigated \cite{Leifer,Barrett,Branciard,Chaturvedi2020quantum,Pusey_2012,ray2024epistemic}. Beyond foundational insights, the relevance of these notions extends to practical applications in quantum information science. Distinguishability and antidistinguishability have been explored within quantum cryptography \cite{JánosABergou_2007, Bae_2015, PhysRevA.98.012330} and have been shown to play crucial roles in achieving quantum advantages in communication tasks \cite{Chaturvedi2020quantum,PhysRevResearch.6.043269, PhysRevLett.115.030504, PhysRevResearch.2.013326,PhysRevLett.125.110402}.

While distinguishing quantum states has been well-established, the problem of distinguishing generic quantum channels is far more complex \cite{sacchi2005, piani2005, pianiPRA, duan2009, harrow2010, piani2015, NakahiraPRL,MuraoPRL2021,ji2024barycentric}. In particular, studies on distinguishing general quantum measurements  (`positive operator-valued measures' (POVMs)) have focused predominantly on the multiple-shot regime, where multiple copies of the measurement are available  \cite{ji2006, ziman2008, ziman2009, fiurasek2009, Pucha_a_2021, ghoreishi2024multipleshotlabelingquantumobservables}. Fewer works explore the more challenging single-shot setting, where a single copy of each measurement is sampled from a given set of measurements \cite{Sedlak2014,PhysRevA.98.042103,njp2021,Nidhin-pra}. However, these studies on measurements presume that the post-measurement states are inaccessible. Moreover, antidistinguishability in the context of measurement remains largely unexplored, except \cite{ji2024barycentric}. 
Interestingly, quantum theory offers four distinct scenarios for investigating the distinguishability (or antidistinguishability) of measurements, depending on two key factors: whether single or entangled systems are used to probe the measurement and whether the post-measurement states are accessible. In this work, we conduct a comprehensive and comparative study of these four scenarios in the single-shot regime.

We formulate the problem of distinguishability and antidistinguishability by evaluating the probability of distinguishing and antidistinguishing quantum measurements, sampled from a given set, in the single-shot regime under the following scenarios: $(i)$ using single systems and without access to post-measurement states, $(ii)$ using entangled systems and without access to post-measurement states, $(iii)$ using single systems with access to post-measurement states, and $(iv)$ using entangled systems with access to post-measurement states. For each scenario, we provide simplified expressions of distinguishability and antidistinguishability of quantum measurements using a generalized formulation of distinguishability and antidistinguishability of quantum states. In scenario $(i)$, we derive a closed-form expression of distinguishability and antidistinguishability of an arbitrary set of qubit projective measurements. For a pair of qubit projective measurements, the optimal value of distinguishability in scenario $(ii)$ is always achieved when the probing system is initially prepared in a maximally entangled state. Compellingly, we found examples of two qubit POVMs for which the non-maximally entangled state outperforms the maximally entangled state as the initial entangled state for distinguishability. It is evident that the distinguishing (or antidistinguishing) probability of any set of measurements in scenario $(iv)$ is always greater than or equal to the probability in any other scenario, and the probability of distinguishability (or antidistinguishability) in scenario $(i)$ is lesser or equal to any other scenario. We show that these hierarchies are strict by providing explicit examples of qubit measurements. Our study reveals that for any pair of qubit projective measurements, scenario $(iii)$ yields better distinguishability than scenario $(ii)$. 
However, the relationship between scenarios $(ii)$ and $(iii)$ is not strictly one-directional; we construct counterexamples to demonstrate that the relative advantage can depend on the specific measurements involved. By introducing novel variants of the well-known \textit{trine} qubit POVM \cite{Sedlak2014}, we construct pairs (and triples) of qubit quantum measurements such that they are perfectly distinguishable (and antidistinguishable) in scenario $(ii)$ but not in scenario $(iii)$, and vice versa. Furthermore, we present examples of qubit measurements that are perfectly distinguishable (and antidistinguishable) only in scenario $(iv)$, highlighting the unique power of entanglement combined with access to post-measurement states.

The paper is organized as follows. In section \ref{sec II}, we present a general formulation of distinguishability and antidistinguishability of quantum states and discuss some already-known results. In section \ref{SEC III}, the problem of distinguishability and antidistinguishability of quantum measurements is introduced. We describe all the aforementioned four different scenarios in detail. Section \ref{SEC new} presents several results that are applicable to arbitrary qubit measurements. In the subsequent section, we introduce variants of the trine POVMs to illustrate the comparative advantages of different scenarios. Through these examples, we demonstrate how certain scenarios enable perfect distinguishability (or antidistinguishability) where others do not. In conclusion, we summarize the key findings and discuss several open problems and potential future research directions.

\section{Distinguishability and antidistinguishability of quantum states} \label{sec II}

Before we dive into quantum measurements, in this section, we give a small overview of the distinguishability and antidistinguishability of quantum states. Assume, we are given $n$ a priori known quantum states $\{\rho_k\}_{k=1}^n$ and they are associated with a set of positive numbers $\{q_k\}_k, q_k >0$.  Distinguishability of $n$ quantum states $\{\rho_k\}_{k=1}^n$ with a set of positive numbers $\{q_k\}_k$, is defined as,
\bea  \label{pD}
\mathcal{DS}[\{\rho_k\}_k,\{q_k\}_k] &=&  \max_{\{M\}}\Bigg\{ \sum_{k}q_k \ p(a=k|\rho_k,M)\Bigg\} \nonumber \\
&=& \max_{\{M\}}\Bigg\{ \sum_{k}q_k \ \tr(\rho_kM_k)\Bigg\} ,
\eea 
where $M$ is a $n$-outcome measurement that distinguish $\{\rho_k\}_{k=1}^n$, and $p(a|\rho_k,M)$ denotes the probability of outcome $a$ when measurement $M$ is implemented on the quantum state $\rho_k$. 
 Note that, the way this distinguishability is defined here, we do not impose $\sum_k q_k=1$, i.e., $\{q_k\}_k$ may not be not a probability distribution, on purpose. In the course of the paper, we will see for most of the scenarios, measurement distinguishability (antidistinguishability) will be related to the problem of distinguishability (antidistinguishability) of the states such that the associated coefficients $\{q_k\}_k$ may not form a probability distribution. Thus, not taking $\{q_k\}_k$ as a probability distribution will help us to simplify the expression of distinguishability (antidistinguishability) of quantum measurements. 
It is easy to see that the maximum value of $\mathcal{DS}$ is $\sum_k q_k$, and it will happen when all the states are perfectly distinguishable.
For two quantum states, Hellstrom \cite{Hellstrom} showed \eqref{pD} reduces to,
\be\label{dspsi12}
\mathcal{DS}[\{\rho_1,\rho_2\},\{q_1,q_2\}] = q_2  + \| q_1\rho_1 - q_2\rho_2 \|,
\ee 
where $\| \cdot \|$ denotes maximum eigenvalue. For two pure states, this becomes,
\bea\label{DSpsi12}
&& \mathcal{DS}[\{\psi_1,\psi_2\},\{q_1,q_2\}] \nonumber \\
&=& \frac12 \left( (q_1+q_2) + \sqrt{(q_1+q_2)^2 - 4 q_1 q_2 |\la \psi_1|\psi_2\ra |^2 }\right).\nonumber\\
\eea  
If $\{q_1,q_2\}$ from a probability distribution, i.e., $q_1+q_2=1$, then \eqref{DSpsi12} will become the famous Hellstrom bound \cite{Hellstrom}.\\

Similarly, one can use the same set-up to antidistinguish $n$ previously given quantum states. 
Antidistinguishability of $n$ quantum states $\{\rho_k\}_{k=1}^n$ is also a linear function of a set positive numbers $\{q_k\}$, which is defined as,
\bea  
\mathcal{AS}[\{\rho_k\}_k,\{q_k\}_k] &= & \max_{\{M\}}\Bigg\{ \sum_{k,a} q_k \ p(a\neq k|\rho_k,M)\Bigg\}. \nonumber
\eea
As $\sum_a (p(a\neq k|\rho_k,M)+p(a=k|\rho_k,M))=1$ , the above expression becomes,
\bea\label{pA}
\mathcal{AS}[\{\rho_k\}_k,\{q_k\}_k] &=& \sum_k q_k - \min_{\{M\}}\Bigg\{\sum_{k}q_k \ p(a=k|\rho_k,M)\Bigg\} \nonumber \\
&=& \sum_k q_k - \min_{\{M\}}\left\{\sum_{k}q_k \ \tr(\rho_k M_k)\right\} .
\eea 
 It is needless to say that, for any two states,
\be 
\mathcal{DS}[\{\rho_1,\rho_2\},\{q_1,q_2\}] = \mathcal{AS}[\{\rho_1,\rho_2\},\{q_1,q_2\}] .
\ee 
The sufficient conditions \cite{Caves02} for perfect antidistinguishability of three pure quantum states, i.e.,
\be \label{ADSpsi123}
\mathcal{AS}[\{\psi_1,\psi_2,\psi_3\},\{q_1,q_2,q_3\}] = \sum_k q_k,
\ee 
are following:
\begin{subequations}\label{condAS}
\be
    x_1 + x_2 + x_3 < 1 
\ee
\be 
(x_1+x_2+x_3-1)^2 \geqslant 4x_1x_2x_3 ,
\ee 
\text{  where $x_1=|\la \psi_1|\psi_2\ra|^2, x_2=|\la \psi_1|\psi_3\ra|^2, x_3=|\la \psi_2|\psi_3\ra|^2$.}\\
\end{subequations}

For the qubit states, the above conditions are necessary and sufficient. In this case, these conditions are equivalent to the fact that the three qubit states lie on a great circle on the Bloch sphere, and the sum of every pair of angles between the vectors is greater than or equal to $\pi$ \cite{ray2024epistemic}.

\section{Distinguishability and Antidistinguishability of  Quantum measurements}\label{SEC III} 

Quantum measurement having $m$  distinct outcome is defined by a set of operators,
\be 
M :=\{F_a\}_a = \{F_1,\cdots,F_m\} , 
\ee 
such that $\sum_{a=1}^m F_a^\dagger F_a = \I$. When this measurement is performed on a quantum state $\rho$,
the probability of obtaining outcome $a$ is given by ,
\be \label{prob}
p(a|\rho,F_a) = \tr(\rho F^\dagger_aF_a) ,
\ee 
and post-measurement state is,
\be
\rho_a=\frac{F_a\rho  F^\dagger_a}{\tr(\rho F^\dagger_aF_a)}.
\ee
Projective measurement is a particular case where $F_a$ are projectors.

Now, we define distinguishability and antidistinguishability of quantum measurements analogous to the similar way we defined quantum states.
We consider a priori known set of $n$ measurements acting on $d$-dimensional quantum states, each having $m$ outcomes, defined by $\{F_{a|x}\}_{a,x}$, which are sampled from a probability distribution $\{p_x\}_x$, i.e., $p_x > 0, \sum_x p_x =1$. Here $x\in\{1,\cdots,n\}$ denotes the measurements and $a\in\{1,\cdots,m\}$ denotes the outcomes. Note that, in general $d$ may not be same as $m$. We define the positive operators, commonly known by POVM elements, as
\be 
M_{a|x} = F^\dagger_{a|x}F_{a|x} ,
\ee
which is often convenient to use, particularly when we do not consider the post-measurement states.  To distinguish or antidistinguish these $n$ measurements $\{F_{a|x}\}_{a,x}$, the measurement device is given a known quantum state, can be single or entangled, and the device carries out one of these $n$ measurements. Measurements and this initial quantum state belong to $\mathcal{C}^d$. After performing one of the measurements, the device gives a classical output and a post-measurement state. Based on the initial quantum state and whether we have access to the post-measurement state or not, we can formulate four different situations: \\




In the next subsections, we give the detailed formulation of all four scenarios for distinguishability and antidistinguishability of quantum measurements.


\begin{widetext}

\begin{figure}
    \centering
    \begin{subfigure}[b]{0.4\textwidth}
        \centering
        \includegraphics[scale=0.2]{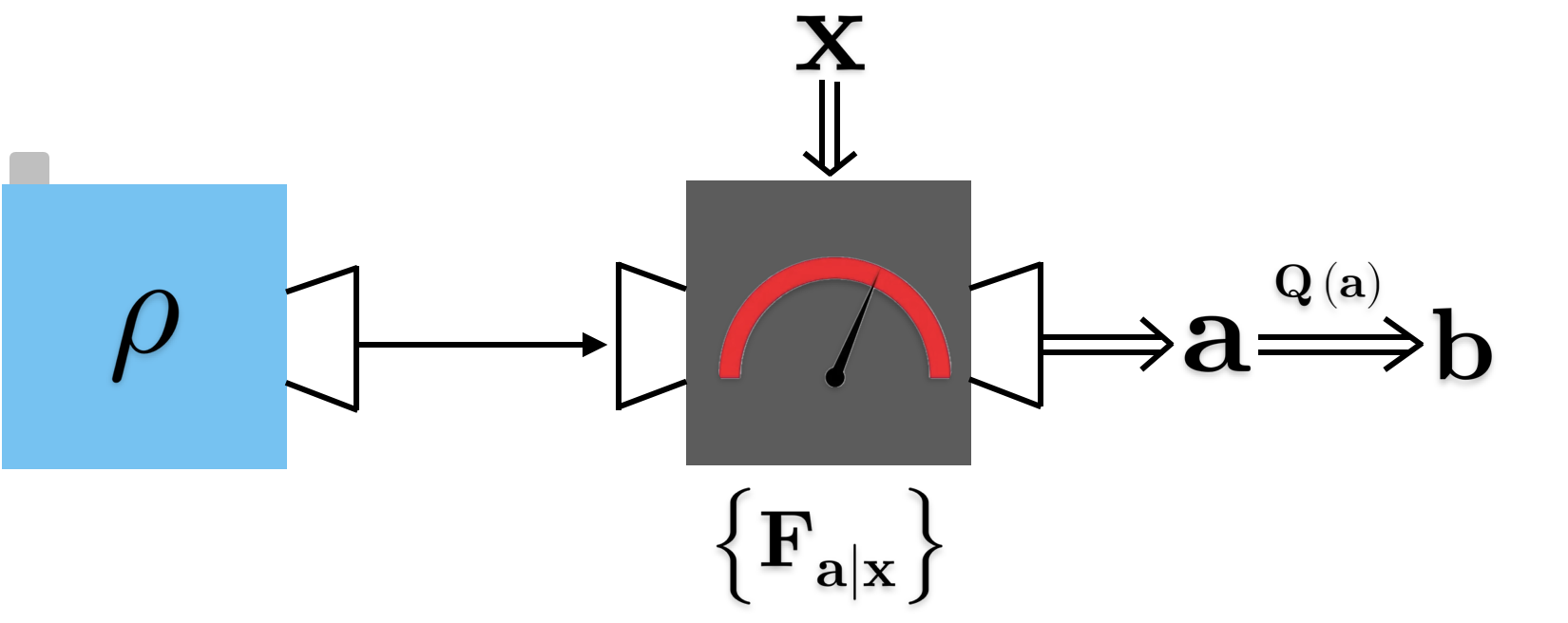}
        \caption{With single system and without post-measurement state ($\mathcal{DMS}/\mathcal{AMS}$)}
        \label{1a}
    \end{subfigure}
    \hfill
    \begin{subfigure}[b]{0.4\textwidth}
        \centering
        \includegraphics[scale=0.2]{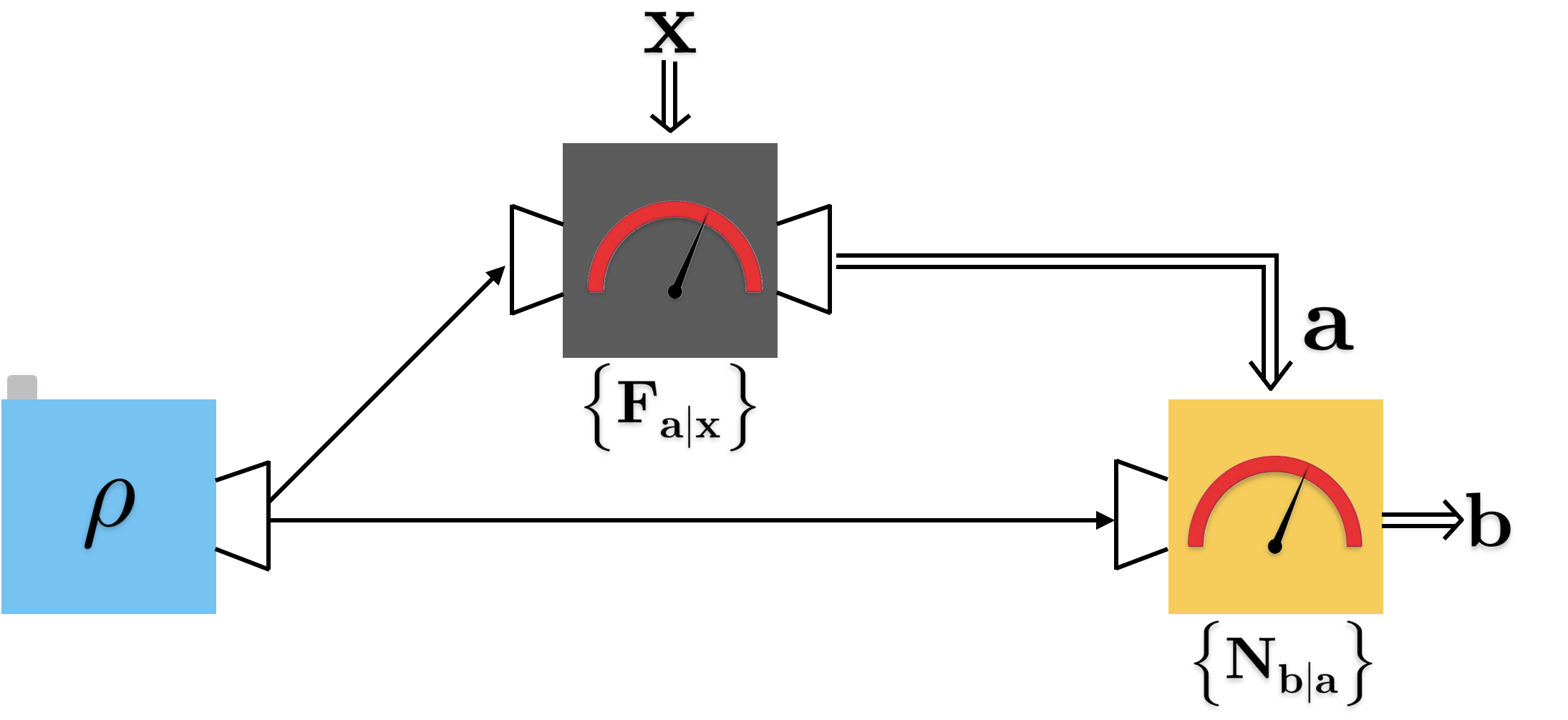}
        \caption{With entangled system and without post-measurement state ($\mathcal{DME}/\mathcal{AME}$)}
        \label{1b}
    \end{subfigure} \\
    \begin{subfigure}[b]{0.4\textwidth}
        \centering
        \includegraphics[scale=0.2]{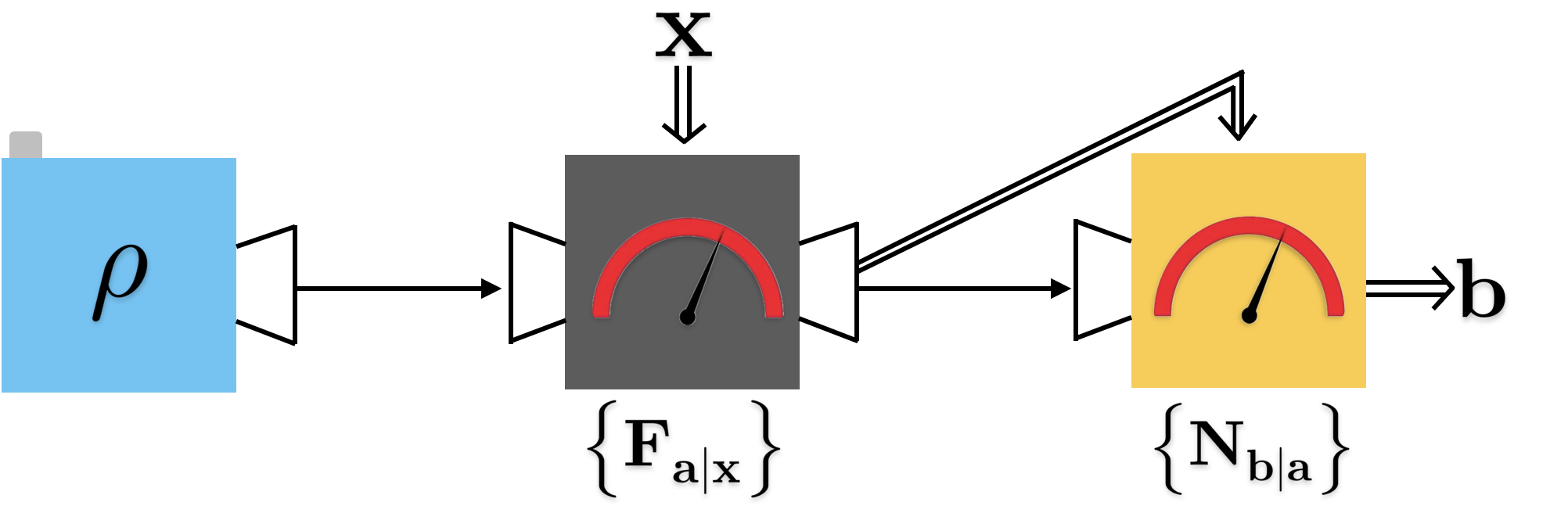}
        \caption{With single system and with post-measurement state ($\mathcal{D\overline{M}S}/\mathcal{A\overline{M}S}$)}
        \label{1c}
    \end{subfigure}
    \hfill
    \begin{subfigure}[b]{0.4\textwidth}
        \centering
        \includegraphics[scale=0.2]{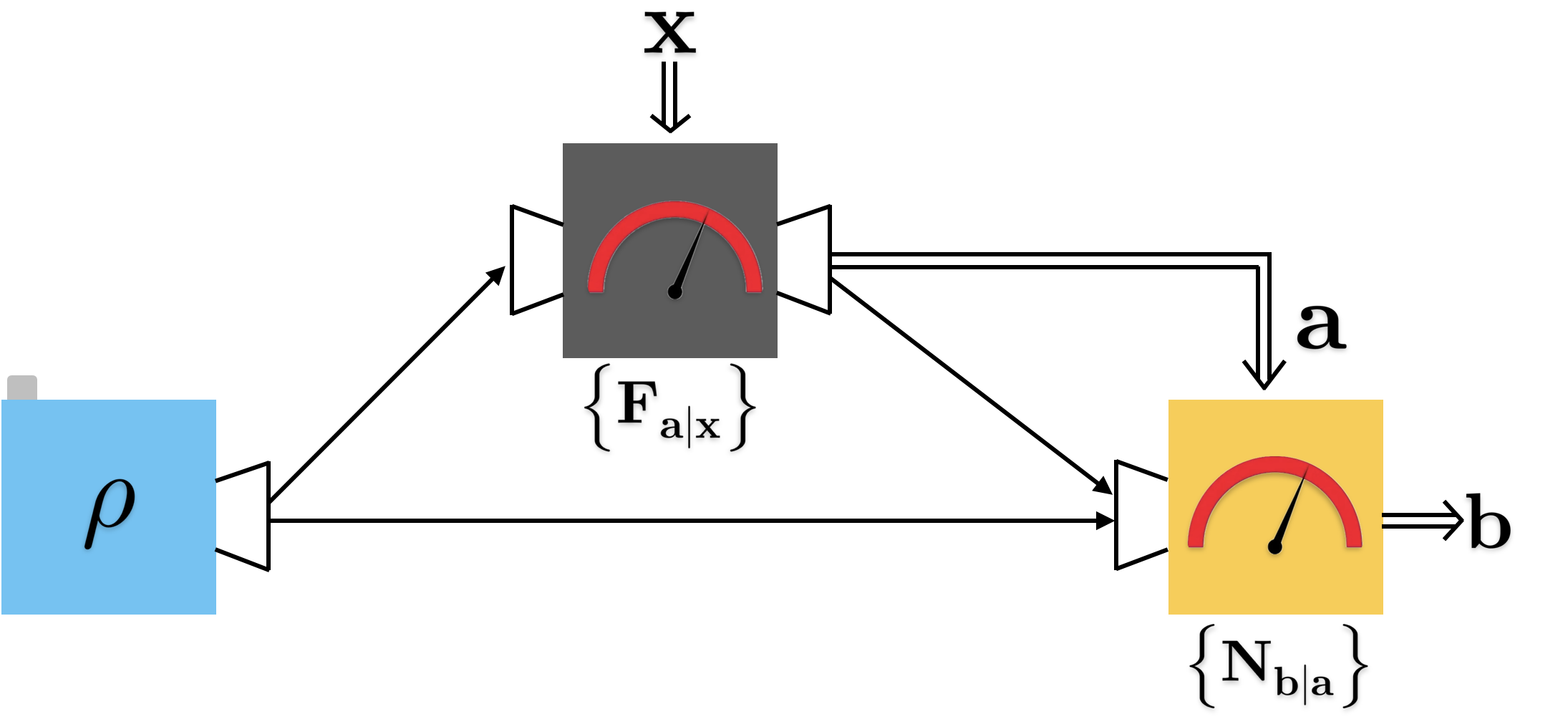}
        \caption{With entangled system and with post-measurement state ($\mathcal{D\overline{M}E}/\mathcal{A\overline{M}E}$)}
        \label{1d}
    \end{subfigure} 
    \caption{Four different scenarios are depicted to distinguishing (and antidistinguishing) a set of given measurements $\{F_{a|x}\}_{a,x}$. Here, double-line arrows denote classical variables, while single-line arrows denote quantum systems. Grey devices implement the measurements we want to distinguish (and antidistinguish). While we have full control over the blue preparation and the yellow measurement devices, we implement the best possible states and measurements in these to accomplish the task. In all the scenarios, outcome $b$ is the answer to distinguish (and antidistinguish).}
    \label{fig}
\end{figure}
\end{widetext}

\subsection{With single system and without the post-measurement state }\label{IIIA}
At first, we can start with a single system as the initial state and only have the classical output of the measurement device. We do not have access to the post-measurement state in this case (see FIG. \eqref{1a}). In this scenario, we will get an outcome with the classical probability $p(a|x,\rho)$. Upon receiving this classical output, in general, one can perform a classical post-processing defined by $p(z|a)$, where $z\in\{1,\cdots,n\}$ and $\forall a$, $\sum_z p(z|a)=1.$ Post-processing protocol acts on $a$ and returns output $z$ for a given $a$ such that $z$ is the guess of the measurement. So, maximum distinguishability of the set of quantum measurements with single systems, denoted by $\mathcal{DMS}$, is given by,
\bea 
&& \mathcal{DMS}\left[\left\{F_{a|x}\right\}_{a,x},\left\{p_x\right\}_x\right] \nonumber \\
&=& \max_{\rho} \sum_{x,a,z} p_x p\left(z= x|a\right) p\left(a|x,\rho\right) \nonumber\\
&=& \max_{\rho} \sum_{a} \max_x \left\{p_x p\left(a=x|x,\rho\right)\right\}, \nonumber\\ 
\eea
  in which the last equality follows by taking the best possible post-processing \cite{njp2021}. Replacing the quantum description, we have,
\bea\label{Drho}
&& \mathcal{DMS}\left[\left\{F_{a|x}\right\}_{a,x},\left\{p_x\right\}_x\right] \nonumber\\
&=&  \max_{\rho} \sum_a \max_x \left\{ p_x \mbox{Tr}\left(\rho M_{a|x} \right) \right\} .
\eea


Similarly, we formalize antidistinguishability of quantum measurements with a single system $\rho$. In the case of antidistinguishability, the aim is to guess which measurement is not implemented in the respective run. Using the same setup depicted in FIG. \eqref{1a}, we can write the expression of antidistinguishability of a set of measurements and denote it by $\mathcal{AMS}$ as follows:  
\bea\label{AMS}
&& \mathcal{AMS}\left[\left\{F_{a|x}\right\}_{a,x},\left\{p_x\right\}_x\right] \nonumber \\ 
 & = & \max_{\rho} \sum_{x,a,z} p_x p\left(z\neq x|a\right) p\left(a|x,\rho\right) \nonumber\\
 &=&  1 - \min_{\rho} \sum_{a} \min_x \left\{p_x p\left(a=x|x,\rho\right)\right\} \nonumber\\ 
 & = & 1 - \min_{\rho} \sum_{a} \min_x\left\{p_x\tr\left(\rho M_{a|x}\right)\right\}.
\eea
The third line follows from the fact that $\sum_z p(z|a)=1$.

\subsection{With entangled system and without the post-measurement state}\label{IIIB}
As shown in FIG \eqref{1b}, in this setup, the distinguishability/antidistinguishability of a set of quantum measurements $\{F_{a|x}\}_{a,x}$, a known quantum bipartite state $\rho^{AB}$ is used by two observers, say, Alice and Bob. One of the measurements from the set is carried out on Alice's subsystem of the entangled system, and an outcome $a$ is received, which is conveyed to Bob. We define the reduced state of Bob when outcome $a$ is obtained by Alice for measurement $x$, 
\be \label{rhoBax}
\rho^B_{a|x} = \tr_A\left[\frac{\left(F_{a|x}\otimes \I\right)\rho^{AB}\left(F^\dagger_{a|x}\otimes \I\right)}{\tr\left( \rho^{AB} \left( M_{a|x} \otimes \I \right) \right)}  \right],
\ee 

where $p(a|x,\rho) = \tr \left[(M_{a|x}\otimes \I)\rho^{AB} \right] $.\\

Depending on Alice's outcome, Bob performs a measurement described by the set of POVM elements $\{N_{b|a}\}$ where $b\in\{1,\cdots,n\}$ is the outcome and $a\in\{1,\cdots,m\}$ is the choice of measurement settings, $\{N_{b|a}\}\geqslant 0, \sum_b N_{b|a}=\I$. For distinguishability, this protocol is successful when Bob's output $b$ will be the same as Alice's input $x$. So, the success probability of this task depends on the joint probability of both Alice's and Bob's measurements. Note that since Bob can choose the best possible measurement on his side, any classical post-processing of the outcome of his measurement can be absorbed within the measurement $\{N_{b|a}\}$. Consequently, distinguishability of quantum measurements with entangled systems, denoted by $\mathcal{DME}$, is written by, 
\bea \label{DME1}
&& \mathcal{DME}\left[\left\{F_{a|x}\right\}_{a,x},\left\{p_x\right\}_x\right] \nonumber\\
&=& \sum_{x,a} p_x p\left(a|x,\rho\right)p\left(b=x|a\right)\nonumber\\
&=& \max_{\rho^{AB}, \{N_{b|a}\}} \ \sum_{x,a} p_x \mbox{Tr} \left[\rho^{AB}(M_{a|x}\otimes N_{b=x|a}\right].
\eea 
Leveraging the expression of joint probability and \eqref{rhoBax}, \eqref{DME1} becomes,
\bea \label{DMEsim}
&&\mathcal{DME}\left[\left\{F_{a|x}\right\}_{a,x},\left\{p_x\right\}_x\right] \nonumber\\ 
&=& \max_{\rho^{AB}} \sum_{a} \max_{\{N_{b|a}\}} \sum_{x} p_x p\left(a|x,\rho\right) \mbox{Tr} \left(\rho^{B}_{a|x} N_{x|a}\right) .
\eea
Interestingly, the expression for every outcome $a$ within the summation is similar to the expression of distinguishability of quantum states defined in \eqref{pD}. using that, \eqref{DMEsim} can be re-expressed as,
\bea\label{DME}
&&\mathcal{DME}\left[\left\{F_{a|x}\right\}_{a,x},\left\{p_x\right\}_x\right] \nonumber\\
&=& \max_{\rho^{AB}} \sum_a \mathcal{DS}\left[\left\{\rho^{B}_{a|x} \right\}_x,\left\{p_x p\left(a|x,\rho\right) \right\}_x\right] .
\eea 
For uniform distribution of the measurements, that means $p_x = 1/n$, we have,
\bea \label{DMEequaln}
&&\mathcal{DME}\left[\left\{F_{a|x}\right\}_{a,x},\left\{p_x\right\}_x\right] \nonumber\\ 
&=& \frac1n \max_{\rho^{AB}} \sum_a \mathcal{DS}\left[\left\{\rho^{B}_{a|x} \right\}_x,\left\{p(a|x,\rho) \right\}_x\right].
\eea 

There exist several measurements, including qubit POVM measurements \cite{Sedlak2014} and high-dimensional projective measurements that are perfectly distinguishable in this scenario but not using single systems. An extended study can be found in \cite{njp2021}.

For antidistinguishability, Alice and Bob execute the same protocol, but the task succeeds when Bob's output $b$ is not equal to Alice's input $x$. Similarly, antidistinguishability of quantum  measurements with entangled systems,

\bea
&& \mathcal{AME}\left[\left\{F_{a|x}\right\}_{a,x},\left\{p_x\right\}_x\right] \nonumber\\
&=& \sum_{x,a} p_x p\left(a|x,\rho\right)p\left(b\neq x|a\right) \nonumber\\
&=& \max_{\rho^{AB}, \{N_{b|a}\}} \ \sum_{x,a,b} p_x \operatorname{Tr} \left[\rho^{AB} \left( M_{a|x} \otimes N_{b\neq x|a}\right)\right] \nonumber\\ 
&=& \max_{\rho^{AB}, \{N_{b|a}\}} \ \sum_{x,a,b} p_x \operatorname{Tr} \left[\rho^{AB} \left( M_{a|x}\otimes(\I - N_{b= x|a})\right)\right] \nonumber\\
&=& \max_{\rho^{AB}} \sum_a \min_{\{N_{b|a}\}} \sum_{x} p_x p(a|x,\rho) \left(1 - \operatorname{Tr} \left(\rho^{B}_{a|x} N_{b=x|a}\right)\right).\nonumber 
\eea

Yet again, the term within the summation over $a$ coincides with the antidistinguishability of quantum states introduced in \eqref{pA}. This implies,
\bea\label{AME}
&& \mathcal{AME}\left[\left\{F_{a|x}\right\}_{a,x},\left\{p_x\right\}_x\right] \nonumber\\
&=& \max_{\rho^{AB}} \sum_a \mathcal{AS}\left[\left\{\rho^{B}_{a|x} \right\}_x,\left\{p_x p\left(a|x,\rho\right) \right\}_x\right].
\eea

In this scenario, the measurements we want to distinguish or antidistinguish are applied on Alice's part. So, the dimensions of the measurements and Alice's state should match. However, the state on Bob's part can be of any dimension. At this point, it is an important question to be asked if the dimension of Bob's state has any constraints depending on the dimension of the measurements. We address this question in theorem \ref{th2} (later in Sec. \ref{SEC new}).

\subsection{With single systems and the post-measurement state}

For discriminating the set of quantum measurements in this scenario, we have a setup similar to the single system in the subsection \ref{IIIA}, but with the post-measurement state (see FIG. \eqref{1c}). After applying any of the measurements from the set $\{F_{a|x}\}_{a,x}$, we will have a classical probability $p(a|x,\rho)$. In addition to that, now we have the access of post-measurement state $\rho_{a|x}$, which can be written by,
\be 
\rho_{a|x} = \frac{F_{a|x}\rho F^\dagger_{a|x}}{p(a|x,\rho)} .
\ee 
Therefore, one can perform any measurement depending on the outcome $a$. Let us describe this measurement by a set of POVM elements $\{N_{b|a}\}$ where $a\in\{1,\cdots,m\}$ denotes the measurement setting and $b\in\{1,\cdots,n\}$ is the outcome. The protocol is successful in distinguishing the measurements if $b$ is the same as $x$. Any classical post-processing of the outcome $b$ can be included within the measurement $\{N_{b|a}\}$. Distinguishability of quantum measurements with single systems with the post-measurement state in this prescription, designated by $\mathcal{D\overline{M}S}$, is given by,
\bea 
&&\mathcal{D\overline{M}S}\left[\left\{F_{a|x}\right\}_{a,x},\left\{p_x\right\}_x\right]\nonumber\\
&=&  \sum_x p_x p\left(b=x|x\right) \nonumber \\
&=& \sum_x p_x \sum_a p\left(a|x,\rho\right)p\left(b=x|a\right) \nonumber \\
&=& \max_{\rho} \sum_{a} \max_{\{N_{b|x}\}} \sum_x p_x \underbrace{\tr\left(\rho M_{a|x}\right)}_{p\left(a|x,\rho\right)} \tr \left(\rho_{a|x} N_{b=x|a}\right) . \nonumber \\
\eea
By \eqref{dspsi12}, it takes the form of distinguishability of the post-measurement states for each outcome $a$, and thus,
\bea\label{DIS}
&&\mathcal{D\overline{M}S}\left[\left\{F_{a|x}\right\}_{a,x},\left\{p_x\right\}_x\right]\nonumber\\
&=& \max_{\rho} \sum_a \mathcal{DS}\left[\left\{{\rho}_{a|x}\right\}_x, \left\{p_x p(a|x,\rho)\right\}_x\right].
\eea
 
The same recipe is applicable for antidistinguishability with a different motive for attainment, which is $b\neq x$. 
Antidistinguishability of quantum measurements with single systems 
is denoted by $\mathcal{A\overline{M}S}$, structured like following:
\begin{widetext}
\bea
 \mathcal{A\overline{M}S}\left[\left\{F_{a|x}\right\}_{a,x},\left\{p_x\right\}_x\right]
&=& \sum_x p_x p\left(b\neq x|x\right) \nonumber \\
&=&  \sum_x p_x \sum_a p\left(a|x,\rho\right)p\left(b\neq x|a\right) \nonumber \\
&=& \sum_{a,x} p_x p\left(a|x,\rho\right)  \left(1 - p(b=x|a)\right) \nonumber \\
&=& \max_{\rho} \left[\sum_{a,x} p_x p\left(a|x,\rho\right) - \sum_a \min_{\{N_{b|a}\}}\sum_x p_x p(a|x,\rho) p(b=x|a)\right] \nonumber \\
&=& \max_{\rho}\sum_a \left[\sum_x p_x p(a|x,\rho)
-\min_{\{N_{b|a}\}} \sum_x  p_x {p\left(a|x,\rho\right)} \underbrace{\operatorname{Tr} \left(\rho_{a|x} N_{b=x|a}\right)}_{p(b=x|a)} \right]. 
\eea
\end{widetext}
With the help of \eqref{pA}, it reduces to the antidistinguishability of the post-measurement state summed over all the outcomes, described by,
\bea\label{AbarMS}
&&\mathcal{A\overline{M}S}\left[\left\{F_{a|x}\right\}_{a,x},\left\{p_x\right\}_x\right]\nonumber\\
&=& \max_{\rho} \sum_a \mathcal{AS}\left[\left\{{\rho}_{a|x}\right\}_x, \left\{p_x p(a|x,\rho)\right\}_x\right].
\eea

\subsection{With entangled systems and the post-measurement state}

For discriminating a set of quantum measurements in this scenario, we have a similar setup as with the entangled system in the subsection \ref{IIIB}, but in this case, the post-measurement state of Alice can be accessed by Bob. (see FIG. \eqref{1d})
 Alice and Bob uses a bipartite entangled state $\rho^{AB}$ to distinguish $\{F_{a|x}\}_{a,x}$. Alice executes an unknown measurement picked from the set. After the measurement, Bob has access to Alice's post-measurement state ($\rho^A_{a|x}$) as well as the reduced state ($\rho_{a|x}^B$) to his side. Depending on this joint state ($\rho_{a|x}^{AB}$), Bob conducts a measurement from the set $\{N_{b|a}\}$. Success probability is the highest when $'b=x'$ holds. This scenario is denoted by $\mathcal{D\overline{M}E}$ and as usual, by leveraging the idea of joint probability, it can be written like the following:
\bea \label{DIE}
&&\mathcal{D\overline{M}E}\left[\left\{F_{a|x}\right\}_{a,x},\left\{p_x\right\}_x\right] \nonumber\\
&=& \sum_x \sum_a p_x p\left(a|x,\rho\right)p\left(b=x|a\right)  \nonumber \\
&=& \max_{\rho^{AB},\{N_{b|a}\}}\sum_a \sum_x p_x p\left(a|x,\rho\right) \tr\left( \rho^{AB}_{a|x} N_{b=x|a} \right) \nonumber \\
&=& \max_{\rho^{AB}} \sum_a \mathcal{DS}\left[\left\{\rho^{AB}_{a|x} \right\}_x, \left\{p_xp(a|x,\rho)\right\}_x\right],
\eea 
where,
\bea \label{rhoABax}
& p(a|x,\rho) = \tr\left( \rho^{AB} \left( M_{a|x} \otimes \I \right) \right), \\  & \rho^{AB}_{a|x} = \frac{\left(F_{a|x}\otimes \I\right)\rho^{AB}\left(F^\dagger_{a|x}\otimes \I\right)}{\tr\left( \rho^{AB} \left( M_{a|x} \otimes \I \right) \right)} .
\eea  
For uniform distribution, $p_x=1/n$,
\bea\label{equaldbarme}
&&\mathcal{D\overline{M}E}\left[\left\{F_{a|x}\right\}_{a,x},\left\{p_x\right\}_x\right]\nonumber\\
&=& \frac1n \max_{\rho^{AB}} \sum_a \mathcal{DS}\left[\left\{\rho^{AB}_{a|x} \right\}_x, \left\{p(a|x,\rho)\right\}_x\right] .
\eea

For antidistinguishability, Alice and Bob follow the same procedure with a different condition for success merit, i.e., $b\neq x$. So,
antidistinguishability of quantum measurements in this scenario,
\begin{widetext}
\bea\label{abarmeprev}
\mathcal{A\overline{M}E}\left[\left\{F_{a|x}\right\}_{a,x},\left\{p_x\right\}_x\right] &=& \sum_b \sum_x \sum_a p_x p\left(a|x,\rho\right)p\left(b\neq x|a\right)  \nonumber \\
&=& \sum_{a,x} p_x p\left(a|x,\rho\right)  \left(1 - p(b=x|a)\right) \nonumber \\
&=& \max_{\rho^{AB}}\left[\sum_{a,x} p_x p\left(a|x,\rho\right) - \min_{\{N_{b|a}\}} \sum_{a,x} p_x p\left(a|x,\rho\right)\underbrace{ \tr\left( \rho^{AB}_{a|x} N_{b=x|a} \right)}_{p(b=x|a)}\right]  
\eea
\end{widetext}
Eventually, by \eqref{pA}, \eqref{abarmeprev} shapes into the antidistinguishability of the joint post-measurement states, and it can be written as,
\bea\label{AbarME}
&&\mathcal{A\overline{M}E}\left[\left\{F_{a|x}\right\}_{a,x},\left\{p_x\right\}_x\right] \nonumber\\
&=& \max_{\rho^{AB}} \sum_a \mathcal{AS}\left[\left\{\rho^{AB}_{a|x} \right\}_x, \left\{p_xp(a|x,\rho)\right\}_x\right].
\eea

Similarly to the scenarios of $\mathcal{DME}$ and $\mathcal{AME}$, the sufficient dimension of the state at Bob's side can not be trivially concluded. We discuss this topic in theorem \ref{th2}.





\section{General Results}\label{SEC new}
In this section, we present some generic results regarding the distinguishability and antidistinguishability of the measurements. Let us first point out that, for any set of measurements, the following relation holds,
\begin{align}\label{all}
\mathcal{DMS}(\mathcal{AMS}) & \ \ \leq \ \ \mathcal{D\overline{M}S}(\mathcal{A\overline{M}S}) \notag \\[1em]
&\hspace{-1cm}\rotatebox[origin=c]{90}{$\geq$}  \hspace{2cm} \rotatebox[origin=c]{90}{$\geq$} \notag \\[1em]
\mathcal{DME}(\mathcal{AME}) & \ \ \leq \ \  \mathcal{D\overline{M}E}(\mathcal{A\overline{M}E}).
\end{align}

These implications are straight-forward since any strategy with single systems without the post-measurement state is a particular instance of single systems with the post-measurement state with trivial measurement on the post-measurement state as well as a particular instance with bipartite states without the post-measurement state with product state. Similarly, any strategy with an entangled state and any strategy with a single system and post-measurement state are particular instances of the strategy with an entangled state and post-measurement state.

Though there exists a qualitative analysis, it is very difficult to calculate the closed form of any of these quantities for any set of measurements. For the simplest case, which is a set of qubit projective measurements, we want to find a closed form of $\mathcal{DMS}$ and $\mathcal{AMS}$.

\begin{widetext}
\begin{thm}\label{th1}
    For a set of qubit projective measurements defined as $\{F_{1|x}\}=\{\ket{\psi_x}\bra{\psi_x}\}_{x=1}^n$ and sampled from the probability distribution $\{p_x\}_{x=1}^n$, the following holds:
    \bea\label{dmsqubitprojective}
      &&\mathcal{DMS}\left[\left\{\ket{\psi_x}\right\}_x,\left\{p_x\right\}_x\right]= \max_{x,x^\prime,x\neq x^\prime} \left\{ \frac12 \left( \left(p_x+p_{x^\prime}\right) + \sqrt{(p_x+p_{x^\prime})^2-4p_xp_{x^\prime}|\la \psi_x|\psi_{x^\prime}\ra |^2} \right) \right\},
    \eea
    \bea\label{amsqubitprojective}
      &&\mathcal{AMS}\left[\left\{\ket{\psi_x}\right\}_x,\left\{p_x\right\}_x\right]
= 1-\min_{x,x^\prime,x\neq x^\prime}\left\{\frac12 \left( (p_x+p_{x^\prime}) - \sqrt{(p_x+p_{x^\prime})^2-4p_xp_{x^\prime}|\la \psi_x|\psi_{x^\prime}\ra |^2} \right) \right\}.
\eea
\end{thm}
\begin{proof}
    For $x\in \{ 1,\cdots,n \}$ and $a\in \{1,2\}$, \eqref{Drho} reduces to,

\bea \label{dmsbinary}
&& \mathcal{DMS}\left[\left\{F_{a|x}\right\}_{a,x},\left\{p_x\right\}_x\right] \nonumber \\
&=& \max_{\rho}\bigg[ \max\left\{ p_1 \mbox{Tr}\left(\rho M_{1|1} \right) , \cdots, p_x\mbox{Tr}\left(\rho M_{1|x} \right) \right\}  + \max\left\{ p_1 - p_1 \mbox{Tr}\left(\rho M_{1|1} \right),\cdots ,  p_x - p_x \mbox{Tr}\left(\rho M_{1|x} \right) \right\}  \bigg] \nonumber \\
&=& \max\bigg\{\max_x\{p_x\},\max_{x,x^{\prime},x\neq x^{\prime}}\{ p_{x^{\prime}} + \| p_x M_{1|x} - p_{x^\prime} M_{1|x^{\prime}}  \| \}  \bigg\} ,
\eea 
where $x,x^{\prime}\in\{1,\cdots,n\}$, $M_{a|x}=F^\dagger_{a|x}F_{a|x}$ and $||.||$ denotes the maximum eigenvalue.

If they are qubit projective measurements , defined as $F_{1|x} = \ket{\psi_x}\!\bra{\psi_x}$, it further simplifies to ,
\bea \label{dmsbinary1}
\mathcal{DMS}\left[\left\{\ket{\psi_x}\right\},\left\{p_x\right\}\right] &=& \max\left\{ \max_x\left\{ p_x \right\}, \max_{x, x' \, (x \neq x')}\left\{ \mathcal{DS}\left[\left\{\ket{\psi_x}, \ket{\psi_{x'}}\right\}, \left\{p_x, p_{x'}\right\}\right] \right\} \right\}\nonumber\\
&=& \max\left\{\max_x\left\{p_x\right\},\max_{x,x^\prime,x\neq x^\prime} \left\{ \frac12 \left( (p_x+p_{x^\prime}) + \sqrt{(p_x+p_{x^\prime})^2-4p_xp_{x^\prime}|\la \psi_x|\psi_{x^\prime}\ra |^2} \right) \right\}\right\}.\nonumber\\
\eea
Similarly, \eqref{AMS} becomes,

\bea \label{amsbinary}
&& \mathcal{AMS}\left[\left\{F_{a|x}\right\},\left\{p_x\right\}_x\right] \nonumber \\
&=& 1-\min_{\rho}\left[ \min\left\{ p_1 \mbox{Tr}\left(\rho M_{1|1} \right) , \cdots, p_x\mbox{Tr}\left(\rho M_{1|x} \right) \right\}  + \min\left\{ p_1 - p_1 \mbox{Tr}\left(\rho M_{1|1} \right),\cdots ,  p_x - p_x \mbox{Tr}\left(\rho M_{1|x} \right) \right\}  \right] \nonumber \\
&=&1- \min\left\{\min_x\{p_x\},\min_{x,x^\prime,x\neq x^\prime}\left\{ p_{x^\prime} + \| p_x M_{1|x} - p_{x^\prime} M_{1|x^\prime}  \|_\bigstar \right\}  \right\} .
\eea
Here $||.||_\bigstar$ denotes the minimum eigenvalue and $M_{a|x}=F_{a|x}^\dagger F_{a|x}$.
If the set of measurements are qubit projective measurements, then
\bea \label{amsbinary1}
&&\mathcal{AMS}\left[\left\{F_{a|x}\right\},\left\{p_x\right\}_x\right]\nonumber\\
&=& 1-\min\left\{\min_x\left\{p_x\right\},\min_{x,x^\prime,x\neq x^\prime}\left\{\frac12 \left( (p_x+p_{x^\prime}) - \sqrt{(p_x+p_{x^\prime})^2-4p_xp_{x^\prime}|\la \psi_x|\psi_{x^\prime}\ra |^2} \right) \right\}\right\}.
\eea
For \eqref{dmsqubitprojective}, we can always exclude first max term inside the second bracket as the second max term is always greater than the first one. Similarly, for \eqref{amsqubitprojective}, we exclude the first min term inside the second bracket. Thus, from \eqref{dmsbinary1} and \eqref{amsbinary1}, we arrive at \eqref{dmsqubitprojective} and \eqref{amsqubitprojective} respectively.
\end{proof}
\end{widetext}

For the special case where the measurements are sampled from an equal probability distribution, $p_x=\frac1n, \forall x$, 
\bea\label{dmsproj}
\mathcal{DMS}&=& \max_{x,x^\prime}\left\{\left( \frac1n + \frac1n\sqrt{1-|\la \psi_x|\psi_{x^\prime}\ra |^2} \right) \right\},
\eea
\bea\label{amsprojective}
\mathcal{AMS}&=& 1-\min_{x,x^\prime,x\neq x^\prime}\left\{ \frac1n - \frac1n\sqrt{1-|\la \psi_x|\psi_{x^\prime}\ra |^2}  \right\}.
\eea

It is easy to see that $\mathcal{AMS}\geqslant \mathcal{DMS}$ for qubit projective measurements in general.
From \eqref{dmsproj}, it is evident that two different qubit projective measurements can never be distinguished in this scenario. However, consider the following two qutrit projective measurements such that $\{F_{a|1}\}= \{\ket{0}\bra{0},\ket{1}\bra{1}, \ket{2}\bra{2}$\} and $\{F_{a|2}\}=\{\frac12\ket{1+2}\bra{1+2},\ket{0}\bra{0},\frac12\ket{1-2}{\bra{1-2}}$\}. If we choose $\ket{0}$ as the best initial state, then $\mathcal{DMS}=1.$

Now, we consider the access of the post-measurement state, which is the case of $\mathcal{D\overline{M}S}$, and we take any pair of qubit projective measurements to find the closed form of the distinguishing probability in this scenario.
\begin{thm}\label{q2proj}
    For two qubit projective measurements defined as $\{F_{1|x}\}=\{\ket{\psi_x}\bra{\psi_x}\}_{x=1}^2$ and sampled from equal probability distribution, 
    \be \mathcal{D\overline{M}S}=\frac12+\frac12\sqrt{1-|\la \psi_1|\psi_2\ra|^4}.
    \ee 
\end{thm}
\begin{proof}
    Without loss of generality, we can apply the same unitary $U$ to both sides of two qubit projective measurements. We choose this $U$ such that $U\ket{\psi_1}=\ket{0}$ and $U\ket{\psi_2}=\ket{\phi}$. So the transformed measurements are: $F_{1|1} = \ket{0}\!\bra{0}$ and $F_{1|2} = \ket{\phi}\!\bra{\phi}$ with uniform distribution, $p_1=p_2=1/2$, where $\ket{\phi}=\cos(\frac{\theta}{2})\ket{0}+\sin(\frac{\theta}{2})\ket{1}$. The initial qubit state is taken as  $\rho=\ket{\psi}\bra{\psi}$ where $\ket{\psi}=\cos(\frac{\alpha}{2})\ket{0}+e^{i\beta}\sin(\frac{\alpha}{2})\ket{1}$ and this $\ket{\psi}$ is needed to be optimized. For these two measurements with this initial state, \eqref{DIS} will be,\\
\bea\label{DISprojective}
\mathcal{D\overline{M}S} &=& \frac{1}{2} \max_{\rho} \left[ \mathcal{DS}\Big(\Big\{\ket{0}, \ket{\phi}\Big\}, \left\{q_1, q_2\right\}\Big) \right. \nonumber \\
&& \left. + \mathcal{DS}\left(\left\{\ket{1}, \ket{\phi^{\perp}}\right\}, \left\{q_3, q_4\right\}\right) \right],
\eea
 where,  
 \bea
&& q_1=\cos^2{\frac{\alpha}{2}},\nonumber\\
&& q_2=\cos^2{\frac{\alpha}{2}}\cos^2{\frac{\theta}{2}}+\sin^2{\frac{\alpha}{2}}\sin^2{\frac{\theta}{2}}+\frac{1}{2}\sin{\alpha}\sin{\theta}\cos{\beta},\nonumber\\
&& q_3=\sin^2{\frac{\alpha}{2}},\nonumber\\
&& q_4=\cos^2{\frac{\alpha}{2}}\sin^2{\frac{\theta}{2}}+\sin^2{\frac{\alpha}{2}}\cos^2{\frac{\theta}{2}}-\frac{1}{2}\sin{\alpha}\sin{\theta}\cos{\beta}.\nonumber
 \eea
\\
Using \eqref{DSpsi12}, \eqref{DISprojective} reduces to,  

\bea\label{DISprojective1}
\mathcal{D\overline{M}S} &=& \frac{1}{2} + \frac{1}{4} \left[ \left\{ \left( q_1 + q_2 \right)^2 - 4q_1 q_2 \cos^2\left( \frac{\theta}{2} \right) \right\}^{\frac{1}{2}} \right. \nonumber \\
&& \left. + \left\{ \left( q_3 + q_4 \right)^2 - 4q_3 q_4 \cos^2\left( \frac{\theta}{2} \right) \right\}^{\frac{1}{2}} \right].
\eea

Maximizing this for $\alpha$ and $\beta$, we get the maximum value of \eqref{DISprojective1} at $(\alpha,\beta)=(\frac{\pi}{2}-\frac{\theta}{2},\pi)$ and $(\frac{\theta}{2}-\frac{\pi}{2},0)$. Putting these values of $\alpha$ and $\beta$ into \eqref{DISprojective1}, we attain at,
\bea\label{DISclosed}
\mathcal{D\overline{M}S}=\frac12+\frac12\sqrt{1-|\la 0|\phi\ra|^4},
\eea
which is nothing but $\mathcal{D\overline{M}S}=\frac12+\frac12\sqrt{1-|\la 0|U^\dagger U|\phi\ra|^4}$ and finally with our definition of $U$, we prove theorem \ref{q2proj}.
\end{proof}

The value of $\mathcal{A\overline{M}S}$ for two qubit projective measurements is the same as that of $\mathcal{D\overline{M}S}$ as it eventually encounters the antidistinguishability of two states, and that is identical with the distinguishability of two states. Theorem \ref{q2proj} tells us that the value of $\mathcal{D\overline{M}S}$ is always less than $1$ for two qubit projective measurements unless we deal with two same measurements with outcome label being opposite. So, it would be a nice venture to find out the conditions for which $\mathcal{D\overline{M}S (A\overline{M}S)} = 1$ for any set of measurements.
\begin{thm}\label{th3}
    $\mathcal{D\overline{M}S}$ ($\mathcal{A\overline{M}S})=1$ for any set of $n$ measurements if and only if there exists an initial state $\rho$ such that for every outcome $a$, one of the conditions holds:\\
    $(i)$ the set of post-measurement state $\{{\rho}_{a|x}\}_x$, where $\tr( \rho M_{a|x}) \neq 0$, are distinguishable (antidistinguishable). \\
    $(ii)$ for all $x$, $\tr( \rho M_{a|x}) = 0$.
\end{thm}
\begin{proof}
Taking into account that the maximum value of the general form of distinguishability \eqref{pD} (or antidistinguishability \eqref{pA}) is given by $\sum_k q_k$, the expressions of  $\mathcal{D\overline{M}S}$ ($\mathcal{A\overline{M}S})$ in Eq.~\eqref{DIS} (Eq.~\eqref{AbarMS}) satisfies the following inequality:
\bea
\mathcal{D\overline{M}S} (\mathcal{A\overline{M}S})
&\leq & \max_{\rho} \left(\sum_{a}\sum_x p_x p(a|x,\rho) \right)\nonumber \\
&=& \max_{\rho} \left( \sum_x p_x \left( \sum_{a} p(a|x,\rho) \right) \right)\nonumber \\
&=& 1.
\eea
Moreover, this inequality becomes equality if and only if for all $a$,
\bea \label{P11}
& \mathcal{DS}\left[\left\{{\rho}_{a|x}\right\}_x, \left\{p_x p(a|x,\rho)\right\}_x\right] = \sum_x p_x p(a|x,\rho); \nonumber \\
& \mathcal{AS}\left[\left\{{\rho}_{a|x}\right\}_x, \left\{p_x p(a|x,\rho)\right\}_x\right] = \sum_x p_x p(a|x,\rho) .
\eea
In other words, $\mathcal{D\overline{M}S} (\mathcal{A\overline{M}S})$, reaches the value 1 whenever the set of post-measurement states $\{\rho_{a|x}\}_x$ such that $p(a|x,\rho) = \tr(\rho M_{a|x})\neq0$ is perfectly distinguishable (antidistinguishable). The conditions \eqref{P11} becomes trivial when for all $x$, $p(a|x,\rho) = \tr(\rho M_{a|x}) = 0$.
Henceforth, theorem \ref{th3} is proved.
\end{proof}

We move into the scenarios that allow the initial state to be entangled. Firstly, we have to find the sufficient dimension of the entangled state at Bob's side to have the highest values of $\mathcal{DME}, \mathcal{AME}, \mathcal{D\overline{M}E}$ and $\mathcal{A\overline{M}E}$. The dimension of the state of Alice is trivially the same as the dimension of measurement, but the dimension of the state of Bob cannot be trivially culminated.

\begin{thm}\label{th2}
    For qudit measurements, it is sufficient to consider qudit-qudit entangled states in order to find $\mathcal{DME}, \mathcal{AME}, \mathcal{D\overline{M}E}$ and $\mathcal{A\overline{M}E}$.
\end{thm}
\begin{proof}
     For $\mathcal{DME},\mathcal{AME},\mathcal{D\overline{M}E}$ and $\mathcal{A\overline{M}E}$ of any two qudit measurements, without loss of generality, we can take the best possible entangled state $\ket{\psi}^{AB}\in\mathcal{C}^d\otimes\mathcal{C}^{d'}$. By Schmidt decomposition, we can always write the state $\ket{\psi}^{AB}=\sum_{i=1}^{d} C_i\ket{\eta_i}\ket{\chi_i}$, where $d'\geqslant d$. 
    From the construction of our initial state, clearly, it is enough to take $d=d'$ as any measurement on Bob's reduced state essentially acts on the $d$-dimensional subspace spanned by $\{\ket{\chi_i}\}_{i=1}^d$. Extra $(d'-d)$ number of bases are redundant. The other case, when $d'<d$, is not considerable because we can not write the states with all the bases of $\ket{\eta_i}$. To make this happen, we need to take at least $d=d'$.
    Thus we become sure about the sufficiency of qudit-qudit entangled state for $\mathcal{DME},\mathcal{AME},\mathcal{D\overline{M}E}$ and $\mathcal{A\overline{M}E}$, that is theorem \ref{th2}.
\end{proof}
 We first discuss about $\mathcal{DME}$ with the initial state being a maximally entangled state. A maximally entangled state has a special feature such that it can be written on any basis. Bob can take any basis in his part of the entangled state because the measurement he needs to choose to distinguish the reduced state is in his control. It is an equivalent operation of acting an unitary on the reduced state on Bob's side after Alice's measurement. But, for Alice's side, is there any preferred basis so that the value of $\mathcal{DME}$ is the maximum? We probe this answer for any two general rank-one qubit measurements.
\begin{thm}\label{equientangled} 
For any two general rank-one qubit measurements, all the maximally entangled states are equivalent for evaluation of $\mathcal{DME}$. 
\end{thm}
\begin{proof}
     Let us consider two general rank one n-outcome POVMs such that $A_i=\sqrt{\alpha_i}\ket{\psi_i}\bra{\psi_i}$ and $B_i=\sqrt{\beta_i}\ket{\zeta_i}\bra{\zeta_i}$ where $i \in \{1,\cdots,n\}$, $\ket{\psi_i}=\cos(\frac{\theta_i}{2})\ket{0}+e^{i\Theta_i}\sin(\frac{\theta_i}{2})\ket{1}$ and $\ket{\zeta_i}=\cos(\frac{\omega_i}{2})\ket{0}+e^{i\Omega_i}\sin(\frac{\omega_i}{2})\ket{1}$.  We consider a general maximally entangled state $\ket{\phi}^{AB}=\frac{1}{\sqrt{2}}(\ket{\eta}\ket{0}+\ket{\eta^\perp}\ket{1})$ as the best possible state where $\ket{\eta}=\cos(\frac{\delta}{2})\ket{0}+e^{i\Delta}\sin(\frac{\delta}{2})\ket{1}$.  Using this state and the two measurements, \eqref{DME} appears as,
     \begin{widetext}
    \bea
     \mathcal{DME}&=&\sum_i\mathcal{DS}\Bigg[\Bigg\{\left(\cos\frac{\theta_i}{2}\cos\frac{\delta}{2}+e^{i(\Delta-\Theta_i)}\sin\frac{\theta_i}{2}\sin\frac{\delta}{2}\right)\ket{0}+\left(\cos\frac{\theta_i}{2}\sin\frac{\delta}{2}-e^{i(\Delta-\Theta_i)}\sin\frac{\theta_i}{2}\cos\frac{\delta}{2}\right)\ket{1},\nonumber\\
    &&\left(\cos\frac{\omega_i}{2}\cos\frac{\delta}{2}+e^{i(\Delta-\Omega_i)}\sin\frac{\omega_i}{2}\sin\frac{\delta}{2}\right)\ket{0}+\left(\cos\frac{\omega_i}{2}\sin\frac{\delta}{2}-e^{i(\Delta-\Omega_i)}\sin\frac{\omega_i}{2}\cos\frac{\delta}{2}\right)\ket{1}\Bigg\},
    \left\{\frac{p_1\alpha_i}{2},\frac{(1-p_1)\beta_i}{2}\right\}\Bigg].\nonumber\\
    \eea
     Two measurements $A_i$ and $B_i$ are sampled from the probability distribution $\{p_1,(1-p_1)\}$ respectively. With help of \eqref{DSpsi12}, the above simplifies to,
    \bea
    &&\mathcal{DME}=\left[\frac12\sum_i\left(\frac{p_1\alpha_i}{2}+\frac{(1-p_1)\beta_i}{2}+
    \sqrt{\left(\frac{p_1\alpha_i}{2}+\frac{(1-p_1)\beta_i}{2}\right)^2-\alpha_i\beta_i|\cos\frac{\theta_i}{2}\cos\frac{\omega_i}{2}+e^{i(\Theta_i-\Omega_i)}\sin\frac{\theta_i}{2}\sin\frac{\omega_i}{2}|^2}\right)\right],\nonumber\\
    \eea
    \end{widetext}
    which is not dependant on $\delta$ and $\Delta$, i.e., $\ket{\eta}$. Thus we arrive at theorem \ref{equientangled}.
\end{proof}
The same result also holds for $\mathcal{D\overline{M}E}$.

At this moment, we can compare the value of $\mathcal{D\overline{M}S}$ and $\mathcal{DME}$ for any two qubit projective measurements. For evaluation of $\mathcal{DME}$, we need to choose the best possible initial entangled state. In the next theorem, we first show numerically that the maximally entangled state as the initial entangled state gives the best value of $\mathcal{DME}$ for any pair of qubit projective measurements, then show the superiority of $\mathcal{D\overline{M}S}$ over $\mathcal{DME}$.
\begin{thm}\label{DbarMSgeqDME}
 For any pair of qubit projective measurements sampled from an equal probability distribution, $\mathcal{D\overline{M}S} \geq  \mathcal{DME}$, and equality holds only if the measurements are the same. 
\end{thm}
\begin{proof}
  Without loss of generality, we can take two qubit projective measurements defined by $F_{1|1}=\ket{0}\bra{0}$ and $F_{1|2}=\ket{\phi}\bra{\phi}$ where $\ket\phi=\cos{\frac{\theta}{2}}\ket{0}+\sin{\frac{\theta}{2}}\ket{1}$. Alice and Bob share a bipartite entangled state $\ket{\zeta}^{AB}=\sqrt{p}\ket{\eta}\ket{0}+\sqrt{1-p}\ket{\eta^{\perp}}\ket{1}$ initially, where $\ket{\eta}=\cos\frac{\delta}{2}\ket{0}+e^{i\Delta}\sin\frac{\delta}{2}\ket{1}$. For these two measurements with the initial state $\ket{\zeta}^{AB}$, \eqref{DMEequaln} becomes,
     \begin{widetext}
     \bea\label{dmenumeric}
\mathcal{DME} &=& \frac{1}{2} \left( \mathcal{DS}\left[\left\{\frac{1}{\sqrt{q_1}}\left(\sqrt{p}\cos\frac{\delta}{2}\ket{0} + \sqrt{1-p}\sin\frac{\delta}{2}\ket{1}\right), \frac{1}{\sqrt{q_2}}\left(\sqrt{p}\left(\cos\frac{\theta}{2}\cos\frac{\delta}{2} + e^{i\Delta}\sin\frac{\theta}{2}\sin\frac{\delta}{2}\right)\ket{0}  \right. \right. \right.\right. \nonumber \\
&& \left.\left.\left. + \sqrt{1-p}\left(\cos\frac{\theta}{2}\sin\frac{\delta}{2} - e^{i\Delta}\sin\frac{\theta}{2}\cos\frac{\delta}{2}\right)\ket{1}\right)\right\}, \left\{q_1, q_2\right\}\right] 
 + \mathcal{DS}\left[\left\{\frac{1}{\sqrt{q_3}}\left(\sqrt{p}\sin\frac{\delta}{2}\ket{0} - \sqrt{1-p}\cos\frac{\delta}{2}\ket{1}\right), \right. \right. \nonumber \\
&& \left. \left. \left. \frac{1}{\sqrt{q_4}}\left(\sqrt{p}\left(\sin\frac{\theta}{2}\cos\frac{\delta}{2} - e^{i\Delta}\cos\frac{\theta}{2}\sin\frac{\delta}{2}\right)\ket{0} +  
\sqrt{1-p}\left(\sin\frac{\theta}{2}\sin\frac{\delta}{2} + e^{i\Delta}\cos\frac{\theta}{2}\cos\frac{\delta}{2}\right)\ket{1}\right)\right\}, \left\{q_3, q_4\right\}\right] \right).\nonumber\\
\eea
    \bea
     && q_1=p\cos^2\frac{\delta}{2}+(1-p)\sin^2\frac{\delta}{2}, q_2=p\left|\cos\frac{\theta}{2}\cos\frac{\delta}{2}+e^{i\Delta}\sin\frac{\theta}{2}\sin\frac{\delta}{2}|^2+(1-p)|\cos\frac{\theta}{2}\sin\frac{\delta}{2}-e^{i\Delta}\sin\frac{\theta}{2}\cos\frac{\delta}{2}\right|^2,\nonumber \\
     &&q_3=p\sin^2\frac{\delta}{2}+(1-p)\cos^2\frac{\delta}{2}, 
     q_4=p\left|\sin\frac{\theta}{2}\cos\frac{\delta}{2}-e^{i\Delta}\cos\frac{\theta}{2}\sin\frac{\delta}{2}|^2+(1-p)|\sin\frac{\theta}{2}\sin\frac{\delta}{2}+e^{i\Delta}\cos\frac{\theta}{2}\cos\frac{\delta}{2}\right|^2.\nonumber
   \eea
     
     By simplification with the help of \eqref{DSpsi12}, \eqref{dmenumeric} reduces to,
     \bea\label{dmeproj}
     & \mathcal{DME}=\frac12+\frac14\left(\sqrt{A}+
     \sqrt{B}\right),
     \eea
     where,
    \bea 
    &&A=\left(q_1+q_2\right)^2-4\left|p\cos\frac{\delta}{2}\left(\cos\frac{\theta}{2}\cos\frac{\delta}{2}+e^{i\Delta}\sin\frac{\theta}{2}\sin\frac{\delta}{2}\right)+\left(1-p\right)\sin\frac{\delta}{2}\left(\cos\frac{\theta}{2}\sin\frac{\delta}{2}-e^{i\Delta}\sin\frac{\theta}{2}\cos\frac{\delta}{2}\right)\right|^2,\nonumber \\
    && B=(q_3+q_4)^2-4\left|p\sin\frac{\delta}{2}\left(\sin\frac{\theta}{2}\cos\frac{\delta}{2}-e^{i\Delta}\cos\frac{\theta}{2}\sin\frac{\delta}{2}\right)+(1-p)\cos\frac{\delta}{2}\left(\sin\frac{\theta}{2}\sin\frac{\delta}{2}+e^{i\Delta}\cos\frac{\theta}{2}\cos\frac{\delta}{2}\right)\right|^2.\nonumber
    \eea
      \end{widetext}
      Numerically, we checked the maximum value of \eqref{dmeproj} by fixing the values of $\theta$ in small intervals (that means fixing the measurements) and exhausting all the values of $\delta,\Delta$, where $\delta\in(0,\pi)$ and $\Delta\in(0,2\pi)$. Then we take $p=\frac12$, which corresponds to maximally entangled state and \eqref{dmeproj} reduces to,
    \bea\label{th5}
     \mathcal{DME} =
    \frac{1}{2}+\frac{1}{2}\sqrt{1-\cos^2\big(\frac{\theta}{2}\big)}
    \eea
We see the evaluated values of $\mathcal{DME}$ in these two calculations are the same. From this, we can conclude that taking a maximally entangled state as the initial state is enough for calculating $\mathcal{DME}$ for two qubit projective measurements. But, it is interesting to notice that there are some non-maximally entangled states that can give the same value as the maximally entangled one.

As $\frac{1}{2}+\frac{1}{2}\sqrt{1-\cos^2\big(\frac{\theta}{2}\big)} \leq \frac{1}{2}+\frac{1}{2}\sqrt{1-\cos^4\big(\frac{\theta}{2}\big)}$, we can say from theorem \ref{q2proj} and \eqref{th5} that $\mathcal{D\overline{M}S} \geq \mathcal{DME}$ and equality holds at $\theta=n\pi$, i.e., when the measurements are same.
\end{proof}

\section{Measurements distinguishable (antidistinguishable) in one scenario but not in other(s)}\label{SEC V}
In this section, we specifically focus on some examples of qubit POVMs to show the advantages of one scenario over the other(s) in both the cases of distinguishability and antidistinguishability. We take all the measurements selected from an equal probability distribution ($\{p_x\}_{x=1}^n = \frac1n$) in all the theorems of this section.
For this, first, let us establish the following notation for these two states $\ket{v_{\pm}}= \frac{1}{2}\ket{0}\pm\frac{\sqrt{3}}{2}\ket{1}$ and $\ket{v^\perp_{\pm}}= \frac{\sqrt{3}}{2}\ket{0}\mp\frac{1}{2}\ket{1}$ which we will use repeatedly. For our convenience, let us define some qubit three-outcome POVM. 
\begin{definition}
[Trine] It is defined by $\{F_a\}$, where $`a'$ means the outcome of the measurement \cite{Sedlak2014}:
\bea \label{trine1}
& F_1 = \sqrt{\frac23}  \ket{0}\!\bra{0}, F_2 = \sqrt{\frac23}  \ket{v_+}\!\bra{v_+},\nonumber\\
& F_3 = \sqrt{\frac23}  \ket{v_-}\!\bra{v_-} .
\eea
\end{definition}
\begin{definition}
[Reverse Trine] It is the complement of \eqref{trine1} and it is denoted by $\{G_a\}$:
\bea\label{trine2}
& G_1 = \sqrt{\frac23} \ket{1}\!\bra{1}, G_2 = \sqrt{\frac23} \ket{v^\perp_+}\!\bra{v^\perp_+}, \nonumber\\
& G_3 = \sqrt{\frac23} \ket{v^\perp_-}\!\bra{v^\perp_-} .
\eea 
\end{definition}

 We will define other measurements based on the above POVMs, which will be useful to find the advantages and disadvantages of the different scenarios of measurement distinguishability and antidistinguishability described in section \ref{SEC III}.
\begin{definition}
[Left Asymmetric Trine] It can be made by applying unitary to the left side of each outcome of \eqref{trine1}, and it is denoted by the symbol $\{H_a\}$:
\bea \label{LAT}
& H_1 = \sqrt{\frac23} U_1 \ket{0}\!\bra{0}, H_2 = \sqrt{\frac23} U_2 \ket{v_+}\!\bra{v_+},\nonumber \\
& H_3 = \sqrt{\frac23} U_3 \ket{v_-}\!\bra{v_-} .
\eea
\end{definition}
One can check operating unitary on the left side does not hamper the positive operators, and so the necessary condition of POVM, i.e., $\sum_{i=1}^3 H_i^\dagger H_i=\I$.
\begin{definition}
   [Right Asymmetric Trine] This is nothing but \eqref{trine1} with a unitary $U$ to the right side of the first outcome. We symbolize it by $\{J_a(\theta)\}$:
\bea \label{RAT}
& J_1 = \alpha^{\frac12}\ket{0}\!\bra{0}U,
 J_2= \beta^{\frac12} \ket{v_+}\!\bra{v_+},   J_3 = \gamma^\frac12 \ket{v_-}\!\bra{v_-} , \nonumber \\ 
& \text{ where } U\ket{0}=\ket{\phi(\theta)}=\cos(\frac{\theta}{2})\ket{0} +  \sin(\frac{\theta}{2})\ket{1}\nonumber\\
&\text{ and } -\pi/3<\theta<\pi/3, \theta\neq 0. \nonumber\\
& \alpha=\left(\frac{\frac12}{\cos^2\frac{\theta}{2}-\frac{1}{4}}\right),
\beta=\left(1 - \frac{\alpha}{2}(1+\frac{2}{\sqrt{3}}\sin\theta) \right),\nonumber\\
& \gamma=\left(1 - \frac{\alpha}{2}(1-\frac{2}{\sqrt{3}}\sin\theta) \right).
\eea 

The necessary range of '$\theta$' makes \eqref{RAT} a valid POVM. We exclude '$\theta=0$' because it gives the usual \textit{'trine'} measurement \eqref{trine1}. 
\end{definition}
   Now, we want to exploit both the ploy of \eqref{LAT} and \eqref{RAT}.
   
\begin{definition}
    [Left-right Asymmetric Trine] It is defined by $\{K_a(\theta)\}$:
    \bea \label{LRAT}
& K_1 = \alpha^{\frac12}\ket{0}\!\bra{0}U,
 K_2= \beta^{\frac12} U_2\ket{v_+}\!\bra{v_+},  \nonumber \\ 
& K_3 = \gamma^\frac12 U_3\ket{v_-}\!\bra{v_-} , \nonumber \\ 
& \text{where} -\frac{\pi}{3}<\theta<\frac{\pi}{3} \text{  and }  \theta \neq 0.
\eea 
$U,\alpha,\beta,\gamma$ has the same meaning as in \eqref{RAT}.
\end{definition}

The construction we used in \eqref{LAT},\eqref{RAT}, and \eqref{LRAT} are related to \eqref{trine1}. The same construction one can use with \eqref{trine2} to make different measurements. For our work, we apply just the construction of \eqref{LRAT} to \eqref{trine2} with a little bit of variation.
 \begin{definition}
     [Left-right Asymmetric Reverse Trine] It is denoted by $\{L_a(\mu)\}$:
 \bea \label{LRArT}
& L_1 = a^{\frac12}V_1\ket{1}\!\bra{1},
 L_2= b^{\frac12} \ket{v_+^{\perp}}\!\bra{v_+^{\perp}}V_2,  \nonumber \\ 
& L_3 = c^\frac12 V_3\ket{v_-^{\perp}}\!\bra{v_-^{\perp}} , \nonumber \\ 
& \text{  where  } V_2\ket{v_+^{\perp}}=\ket{\psi(\mu)}= \cos(\frac{\mu}{2})\ket{0} +  \sin(\frac{\mu}{2})\ket{1}, \nonumber\\
&  \frac{4\pi}{3}<\mu< 2\pi.\nonumber\\
& a=b\left(\cos\mu-\frac{1}{\sqrt{3}}\sin\mu\right),
b=\left(\frac{1}{\cos^2\frac{\mu}{2}-\sqrt{3}\cos\frac{\mu}{2}\sin\frac{\mu}{2}}\right),\nonumber\\
& c=-b\left(\frac{4}{\sqrt{3}}\cos\frac{\mu}{2}\sin\frac{\mu}{2}\right)
\eea 
 \end{definition}
Pictorial descriptions of all these measurements are presented in figure \ref{figure}. There can be many measurements employing these kinds of modifications regarding \eqref{trine1} and \eqref{trine2}. For our paper, these will be used to derive our results. For distinguishability or antidistinguishability of measurements, labeling of the outcomes is really important. In the course of this paper, there will be examples where we will use these measurements but not necessarily in the same order of the outcomes. In those cases, we give the order of the outcomes explicitly.
\begin{widetext}

\begin{figure}[h!]
    \centering
    \begin{subfigure}[b]{0.4\textwidth}
        \centering
        \includegraphics[scale=0.2]{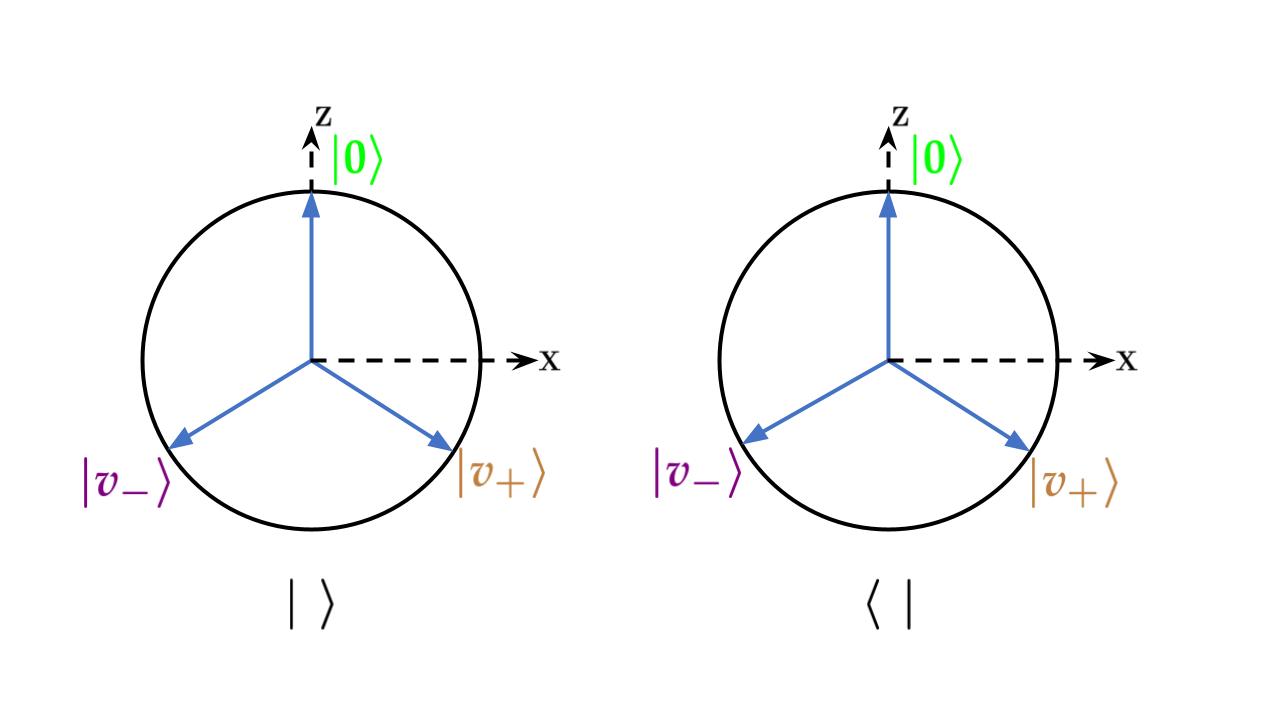}
        \caption{Trine}
        \label{2a}
    \end{subfigure}
    \hfill
    \begin{subfigure}[b]{0.5\textwidth}
        \centering
        \includegraphics[scale=0.2]{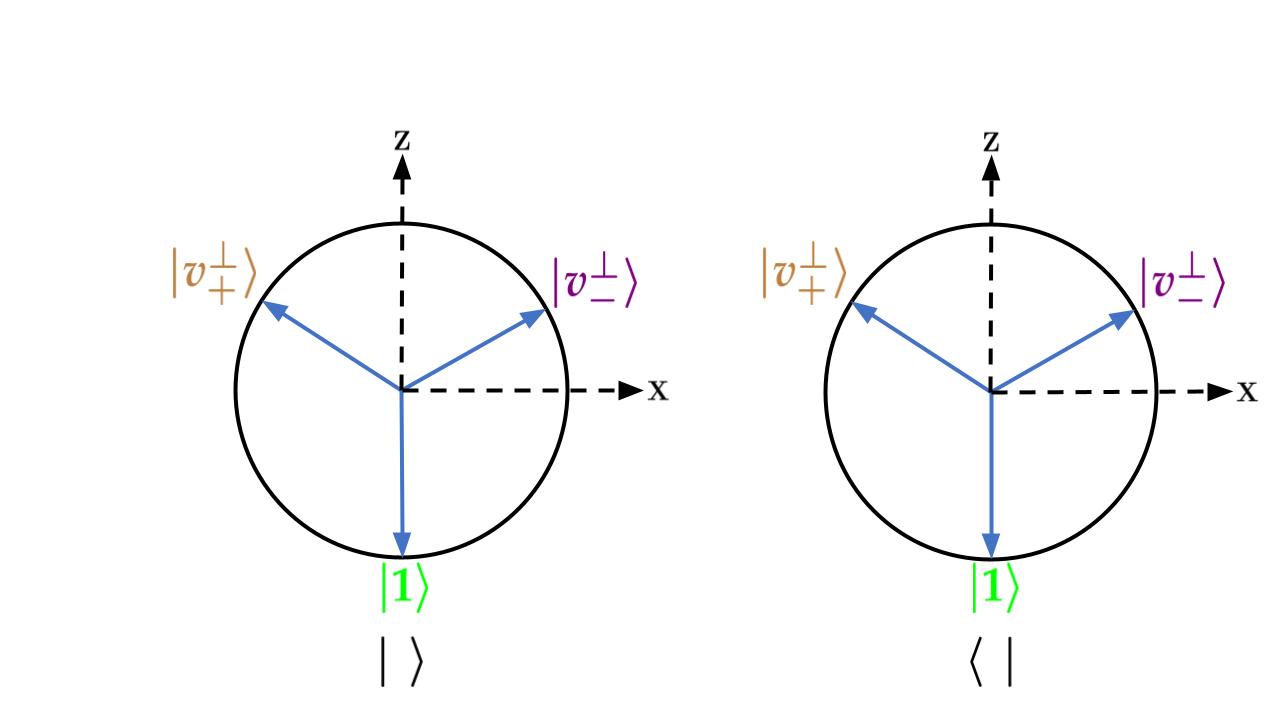}
        \caption{Reverse Trine}
        \label{2b}
    \end{subfigure} \\
    \begin{subfigure}[b]{0.4\textwidth}
        \centering
        \includegraphics[scale=0.2]{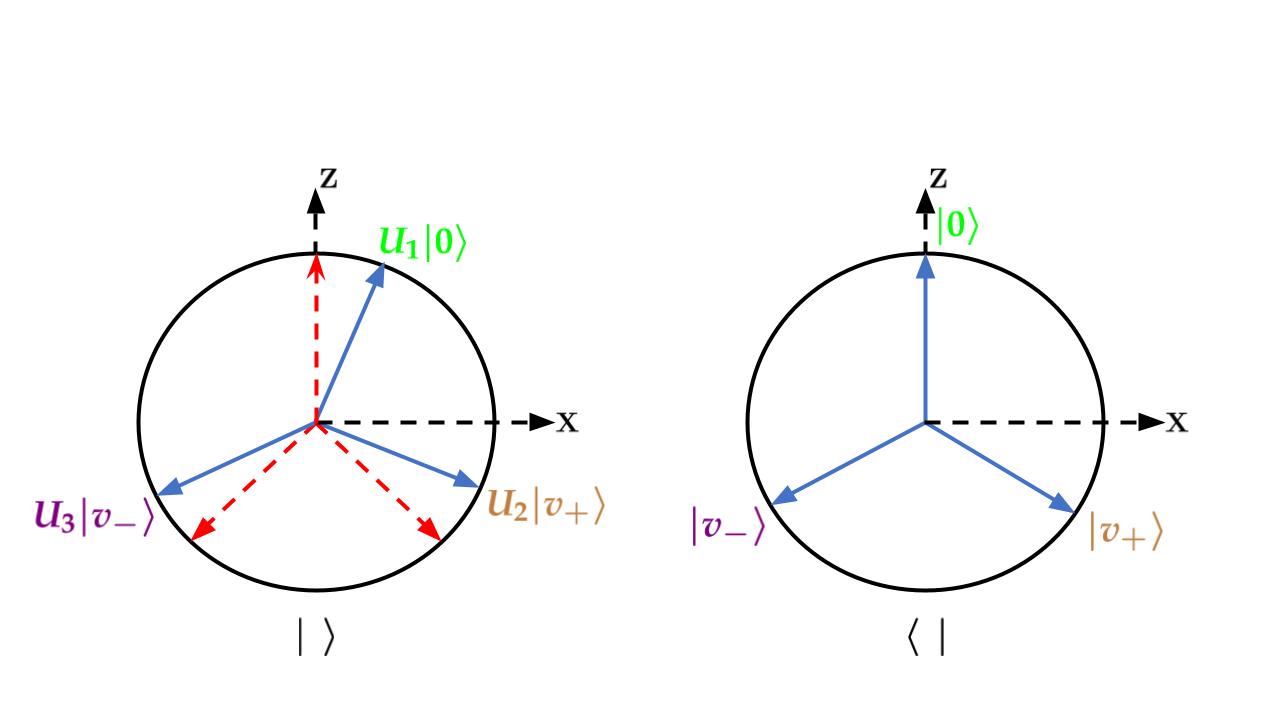}
        \caption{Left Asymmetric Trine}
        \label{2c}
    \end{subfigure}
    \hfill
    \begin{subfigure}[b]{0.45\textwidth}
        \centering
        \includegraphics[scale=0.2]{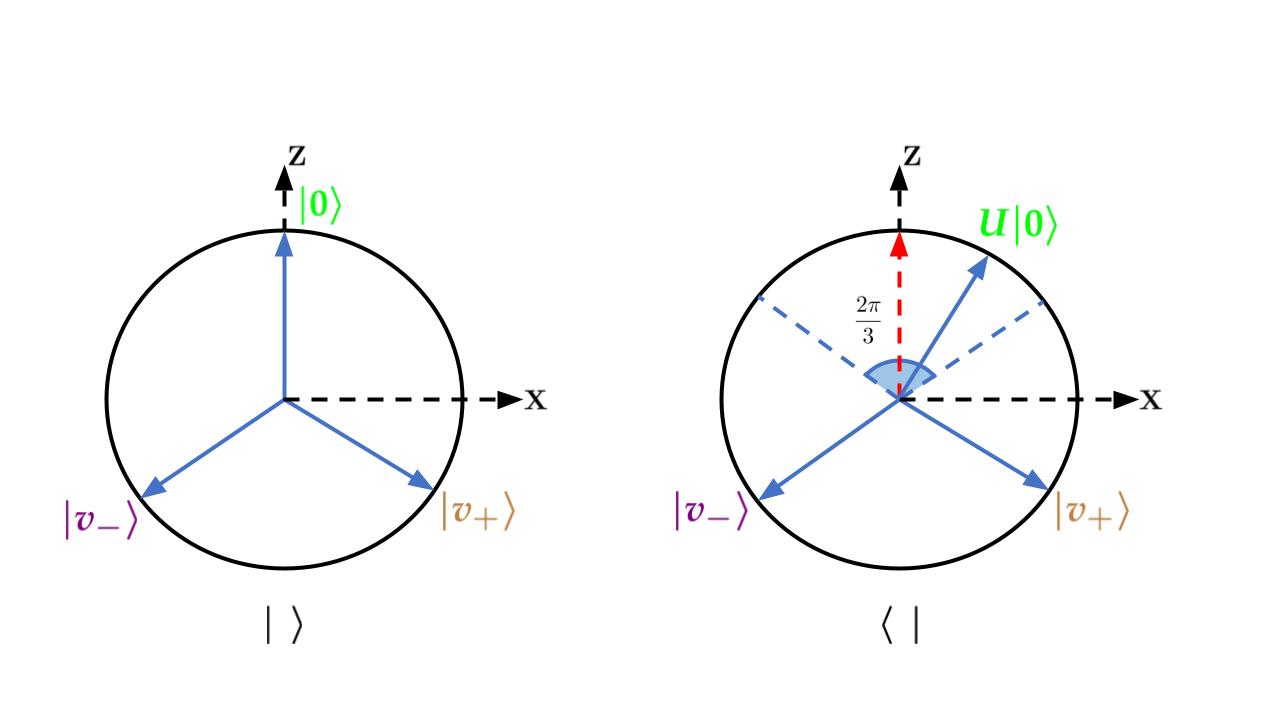}
        \caption{Right Asymmetric Trine}
        \label{2d}
    \end{subfigure} \\
     \begin{subfigure}[b]{0.4\textwidth}
        \centering
        \includegraphics[scale=0.2]{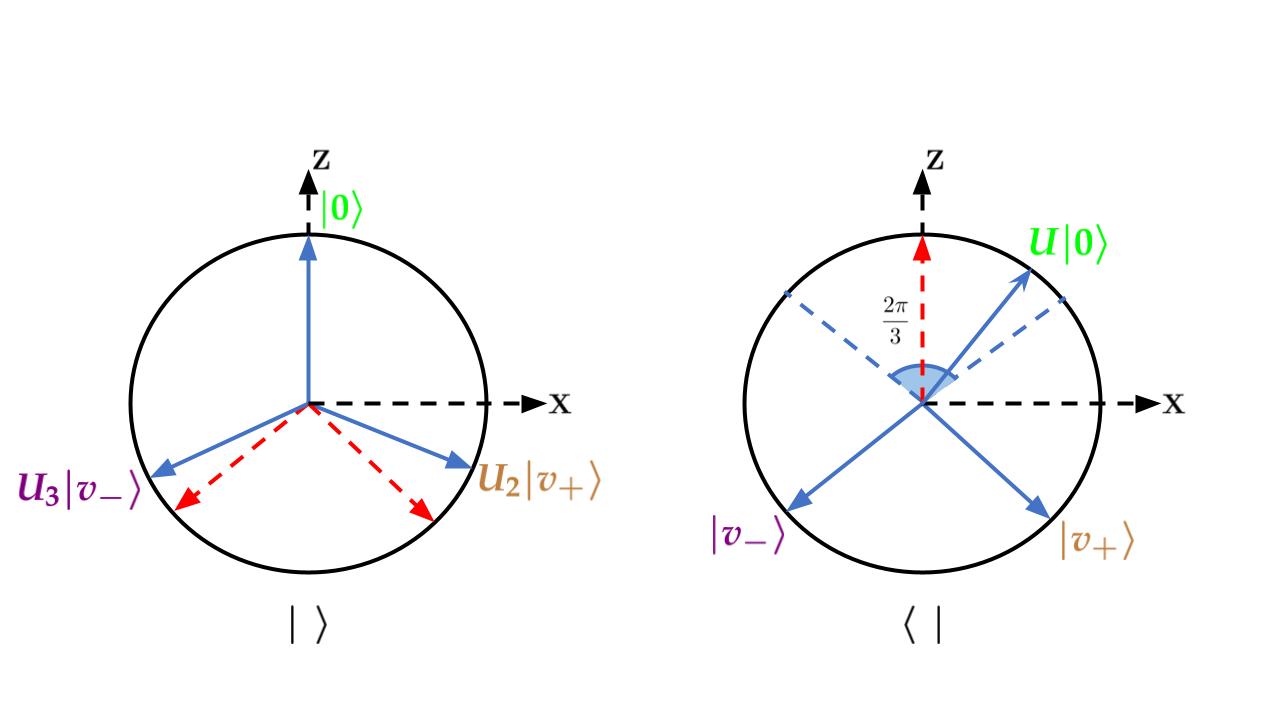}
        \caption{Left-right Asymmetric Trine}
        \label{2e}
    \end{subfigure}
    \hfill
    \begin{subfigure}[b]{0.45\textwidth}
        \centering
        \includegraphics[scale=0.2]{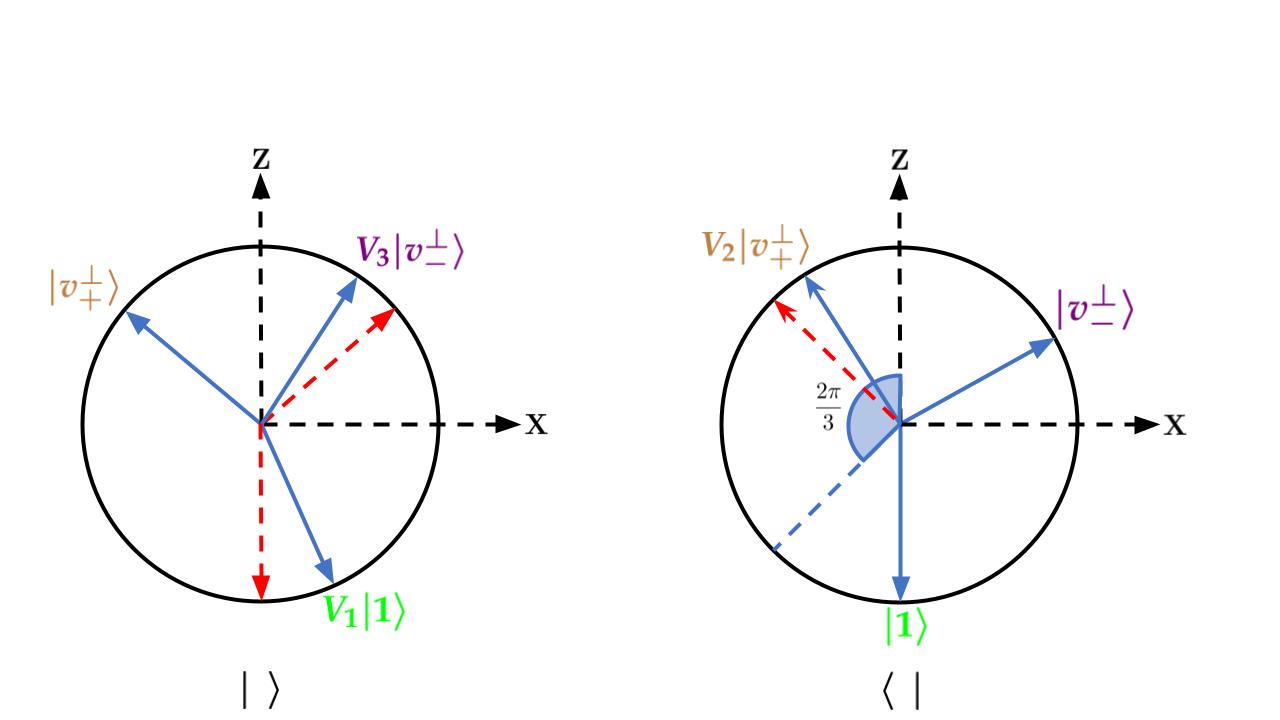}
        \caption{Left-right Asymmetric Reverse Trine}
        \label{2f}
    \end{subfigure} 
    \caption{Pictorial depiction of all the measurements defined at the beginning of Section \ref{SEC V} is provided in the X-Z plane of Bloch sphere. The POVM elements are defined as $\kappa|\cdot\rangle\langle \cdot |$ such that the corresponding ket vector $|\cdot\rangle$ and the bra vector $\langle\cdot|$ are presented, where $\kappa$ is a constant. Each POVM has three outcomes, and different colours are used to write the respective ket and bra vectors for different outcomes. The blue line denotes the state by which POVM elements are made. The red dotted line gives the original position of the `Trine' [in (c),(d),(e)] or the `Reverse Trine' [in (f)]. For the last three images, the valid range of angles in which the bra vectors must lie in order to form a valid POVM is highlighted. }
    \label{figure}
\end{figure}
\end{widetext}

\begin{thm}\label{theorem6}
     The value of $\mathcal{DMS}$ for the measurements \textit{'Right Asymmetric Trine'}\eqref{RAT} and \textit{'Reverse Trine'}\eqref{trine2} is always less than $1$.
\end{thm}
\begin{proof}

Without loss of generality, we can consider the form of starting best state to be $\ket{\zeta}=\sin{\delta}\ket{0}+\cos{\delta}\ket{1}$, where $\delta\in[0,\pi]$. Therefore, with these measurements, from \eqref{Drho}, we can write,
\begin{widetext}
\bea\label{DMSth6}
  \mathcal{DMS}
 &=&\frac{1}{2}\max_\delta\Bigg[\max\left\{\alpha\sin^2\left({\delta+\frac{\theta}{2}}\right),\frac{2}{3}\cos^2{\delta}\right\}
 +\max\left\{\frac{\beta}{4}\left(\sin{\delta}+\sqrt{3}\cos{\delta}\right)^2,\frac{1}{6}\left(\sqrt{3}\sin{\delta}-\cos{\delta}\right)^2\right\}\nonumber\\
 &&+ \max\left\{\frac{\gamma}{4}\left(\sin{\delta}-\sqrt{3}\cos{\delta}\right)^2,\frac{1}{6}\left(\sqrt{3}\sin{\delta}+\cos{\delta}\right)^2\right\}\Bigg].
\eea
\end{widetext}
It is easy to see that for $-\frac{\pi}{3}<\theta< 0$, 
 $\delta=\frac{\pi}{6}$ gives the maximum value of \eqref{DMSth6} and for $0<\theta<\frac{\pi}{3}$, $\delta=\frac{5\pi}{6}$ gives the optimum. With these values of $\delta$, the maximum value of \eqref{DMSth6} reduces to,
\be \label{DMSMAX}
    \mathcal{DMS}=
 \begin{cases}
   \frac{1+\beta}{2}, & \text{for $-\frac{\pi}{3}<\theta< 0$}\\
    \frac{1+\gamma}{2},      & \text{for $0<\theta<\frac{\pi}{3}$}.
\end{cases}
\ee 
As $\beta$ and $\gamma$ are always less than $1$ from our definition, \eqref{DMSMAX} is always less than $1$, which is nothing but \eqref{theorem6}.
\end{proof}

From Theorem \ref{DbarMSgeqDME}, we already know that for a pair of qubit projective measurements, the maximally entangled state is the optimal one for the evaluation of $\mathcal{DME}$. 
  Now we want to check if this conclusion also holds for two rank-one qubit POVMs.

\begin{thm}\label{th10}

     For $\textit{'Right Asymmetric Trine'}$\eqref{RAT} and $\textit{'Reverse Trine'}$\eqref{trine2}, there exists a non-maximally entangled state which gives a better value of $\mathcal{DME}$ than with maximally entangled state for $\theta\in(\frac{\pi}{7},\frac{\pi}{3})\cup(-\frac{\pi}{3},-\frac{\pi}{7})$. 
\end{thm}
\begin{proof}  
As theorem \ref{equientangled} dictates, we can start with the Bell state $\ket{\phi^+}^{AB}=\frac{1}{\sqrt{2}}(\ket{00}+\ket{11})$ for $\mathcal{DME}$ instead of any maximally entangled state,
$\mathcal{DME}$ between \eqref{RAT} and \eqref{trine2} taken from equal probability distribution is given by,
 
\bea
\mathcal{DME}{\{J_a,G_a\}}_{\ket{\phi^+}}
&=&\frac12\left(\mathcal{DS}\left[\left\{U\ket{0},\ket{1}\right\},\left\{\frac{\alpha}{2},\frac13\right\}\right]\right.\nonumber\\
     &&+ \mathcal{DS}\left[\left\{\ket{v_+},\ket{v_+^{\perp}}\right\},
     \left\{\frac{\beta}{2},\frac13\right\}\right]\nonumber\\
     &&\left. + \mathcal{DS}\left[\left\{\ket{v_-},\ket{v_-^{\perp}}\right\},\left\{\frac{\gamma}{2},\frac13 \right\}\right]\right).\nonumber\\
\eea
Leveraging \eqref{DSpsi12},

\bea\label{DMERATmaxent}
&&\mathcal{DME}{\{J_a,G_a\}}_{\ket{\phi^+}}\nonumber\\
     &=&\frac{2}{3}+\frac{\beta+\gamma}{8}+\frac14\sqrt{\left(\frac{\alpha}{2}+\frac13\right)^2-\frac{2\alpha}{3}\sin^2\frac{\theta}{2}}\nonumber\\
     &=&\frac{11}{12}-\frac{\alpha}{8}+\frac14\sqrt{\left(\frac{\alpha}{2}+\frac13\right)^2-\frac{2\alpha}{3}\sin^2\frac{\theta}{2}}.
\eea
Now we will calculate $\mathcal{DME}$ with non-maximally entangled states. We take the most general qubit entangled state, i.e, $\ket{\phi}^{AB}=a\ket{00}+b\ket{01}+c\ket{10}+d\ket{11}$ as the best possible state where $|a|^2+|b|^2+|c|^2+|d|^2 =1$ by normalization. With this $\ket{\phi}^{AB}$ and these two measurements \eqref{RAT} and \eqref{trine2}, we obtain from \eqref{DMEequaln},
 \begin{widetext}  
   \bea\label{dme<1}
   &&\mathcal{DME}{\{J_a,G_a\}}_{\ket{\phi}^{AB}}\nonumber\\
   &=&\frac12\Bigg(\mathcal{DS}\left[\left\{\frac{\sqrt{\alpha}}{\sqrt{p_1}}\left(\left(a\cos\frac{\theta}{2}+c\sin\frac{\theta}{2}\right)\ket{0}+\left(b\cos\frac{\theta}{2}+d\sin\frac{\theta}{2}\right)\ket{1}\right),\frac{\sqrt{2/3}}{\sqrt{p_2}}\left(c\ket{0}+d\ket{1}\right)\right\},\left\{p_1,p_2\right\}\right]\nonumber\\
   &&+ \mathcal{DS}\left[\left\{\frac{\sqrt{\beta}}{\sqrt{p_3}}\left(\left(\frac{a}{2}+\frac{\sqrt{3}c}{2}\right)\ket{0}+\left(\frac{b}{2}+\frac{\sqrt{3}d}{2}\right)\ket{1}\right),\frac{\sqrt{2/3}}{\sqrt{p_4}}\left(\left(\frac{\sqrt{3}a}{2}-\frac{c}{2}\right)\ket{0}+\left(\frac{\sqrt{3}b}{2}-\frac{d}{2}\right)\ket{1}\right)\right\},
    \left\{p_3,p_4\right\}\right] \nonumber\\
    &&+ \mathcal{DS}\left[\left\{\frac{\sqrt{\gamma}}{\sqrt{p_5}}\left(\left(\frac{a}{2}-\frac{\sqrt{3}c}{2}\right)\ket{0}+\left(\frac{b}{2}-\frac{\sqrt{3}d}{2}\right)\ket{1}\right),\frac{\sqrt{2/3}}{\sqrt{p_6}}\left(\left(\frac{\sqrt{3}a}{2}+\frac{c}{2}\right)\ket{0}+\left(\frac{\sqrt{3}b}{2}+\frac{d}{2}\right)\ket{1}\right)\right\},
    \{p_5,p_6\}\right]\Bigg).\nonumber\\
    \eea

    All the probability terms are given below: 
  \bea
  &&p_1=\alpha\left[\left(a\cos\frac{\theta}{2}+c\sin\frac{\theta}{2}\right)^2+\left(b\cos\frac{\theta}{2}+d\sin\frac{\theta}{2}\right)^2\right],
  p_2=\frac23\left[c^2+d^2\right],\nonumber\\
    &&p_3=\beta\left[\left(\frac{a}{2}+\frac{\sqrt{3}c}{2}\right)^2+\left(\frac{b}{2}+\frac{\sqrt{3}d}{2}\right)^2\right],
    p_4=\frac23\left[\left(\frac{\sqrt{3}a}{2}-\frac{c}{2}\right)^2+\left(\frac{\sqrt{3}b}{2}-\frac{d}{2}\right)^2\right],\nonumber\\
    &&p_5= \gamma\left[\left(\frac{a}{2}-\frac{\sqrt{3}c}{2}\right)^2+\left(\frac{b}{2}-\frac{\sqrt{3}d}{2}\right)^2\right],
    p_6=\frac23\left[\left(\frac{\sqrt{3}a}{2}+\frac{c}{2}\right)^2+\left(\frac{\sqrt{3}b}{2}+\frac{d}{2}\right)^2\right].\nonumber
    \eea
   \end{widetext}
  The coefficients are chosen like the following:
\bea
&& a=\cos^2\left(\frac{\pi}{4}-\frac{\theta^2}{8}\right),b=c=0,d=\sin^2\left(\frac{\pi}{4}-\frac{\theta^2}{8}\right).\nonumber
\eea
With these coefficients, one can check, in the given range of $\theta$ in theorem \ref{th10},
\be
\mathcal{DME}{\{J_a,G_a\}}_{\ket{\phi^+}}< \mathcal{DME}{\{J_a,G_a\}}_{\ket{\phi}^{AB}}\nonumber
\ee
up to the fourth decimal.
\end{proof}
A detailed numerical analysis is discussed here.
Outside of the range of $\theta$ taken in theorem \ref{th10}, our chosen coefficients are of no good. From \eqref{DMERATmaxent} and \eqref{dme<1}, we have the expressions of $\mathcal{DME}$ of these two measurements with maximally entangled state and non-maximally entangled state, respectively. For numerical analysis, we fix the value of $\theta$ within the range $\theta\in(-\pi/3,\pi/3)$ in small interval and we iterate all possible values of $a,b,c,d$ of \eqref{dme<1} to get the highest value of  $\mathcal{DME}{\{J_a,G_a\}}_{\ket{\phi}^{AB}}$. For the maximally entangled state, we have a closed form as a function of $\theta$. Then we plot these two $\mathcal{DME}$ against $\theta$, which is shown in figure \ref{DMEplot}.

\begin{figure}[h!]
\includegraphics[width=\linewidth]{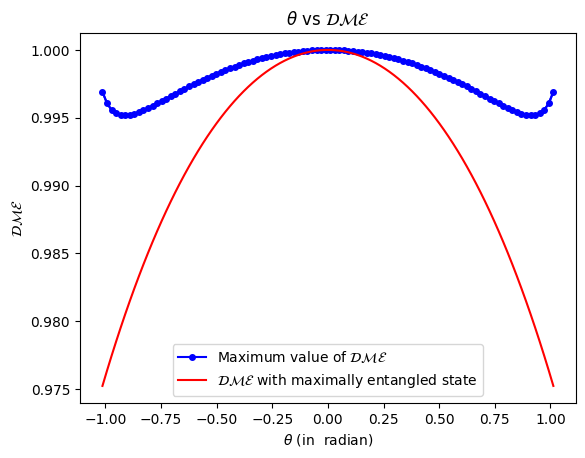}
\caption{$\mathcal{DME}$ between '\textit{Right Asymmetric Trine}' \eqref{RAT} and '\textit{Reverse Trine}' \eqref{trine2} is presented here, once, starting with maximally entangled state and again with non-maximally entangled state. We see a clear gap between the red and the blue scattered line, and that indicates that the non-maximally entangled state gives a better value of $\mathcal{DME}$ than the maximally entangled state.}  
\label{DMEplot}
\end{figure}

We are moving into some examples for which the supremacy of one scenario over the other(s) becomes more clear in the case of distinguishability. 
\begin{thm}\label{trineproof}
    For the two measurements defined by '\textit{Left Asymmetric Trine}' \eqref{LAT} and '\textit{Reverse Trine}'\eqref{trine2}, $\mathcal{D\overline{M}S} < \mathcal{DME}=1$ if and only if at least two of the unitaries $\{U_i\}, i=1,2,3$ are not $\I$.
\end{thm}
\begin{proof}
    We take our best initial state as $\ket{\zeta}$. By substituting these two measurements \eqref{LAT} and \eqref{trine2} into \eqref{DIS}, we get, \\   
    \begin{widetext}
\bea\label{trineDbarMS}
 \mathcal{D\overline{M}S} 
&=& \frac{1}{2} \max_{\ket{\zeta}} \left[\mathcal{DS}\left(\Big\{U_1\ket{0}, \ket{1}\Big\},\left\{\frac{2}{3}|\langle \zeta|0\rangle|^2,\frac{2}{3}|\langle \zeta|1\rangle|^2\right\}\right)
 + \mathcal{DS}\left(\left\{U_2\ket{v_+}, \ket{v^\perp_+}\right\},\left\{\frac{2}{3}|\langle \zeta|v_+\rangle|^2,\frac{2}{3}|\langle \zeta|v^\perp_+\rangle|^2\right\}\right)\right. \nonumber\\
&&\left. + \mathcal{DS}\left(\left\{U_3\ket{v_-}, \ket{v^\perp_-}\right\},\left\{\frac{2}{3}|\langle \zeta|v_-\rangle|^2,\frac{2}{3}|\langle \zeta|v^\perp_-\rangle|^2\right\}\right)\right].\nonumber\\
\eea
 \end{widetext}
 If one of $U_i$=$U_j\neq\I$, we can choose $\ket{\psi}$ such that the probability of getting the outcome $j$ is zero. Now \eqref{trineDbarMS} consists of two pairs of orthogonal states, and for $j'$th outcome, there is only one state with non-zero probability. So, for the $j'$th outcome, it satisfies $(ii)$ of theorem \ref{th3} and for the other two outcomes, it satisfies $(i)$ of theorem \ref{th3}.
 
Now, let us take any two of the unitaries that are not equal to $\I$. Unlike the previous situation, it consists of at least one pair of non-orthogonal states with non-zero probabilities. This contradicts $(i)$ of theorem \ref{th3}.
Therefore, if we choose at least two unitaries that are not $\I$, $\mathcal{D\overline{M}S}<1$.
   
  Applying these two measurements to Alice's part of the maximally entangled state $\ket{\phi^+}=\frac{1}{\sqrt{2}}(|00\rangle+|11\rangle)$,
  Bob's part is projected onto the orthogonal states for each outcome.
  From \eqref{DMEequaln}, it follows,
  \bea\label{dmelatt2}
 \mathcal{DME} &=& \frac12 \left[\mathcal{DS}\left(\Big\{\ket{0}, \ket{1}\Big\},\left\{\frac13 ,\frac{1}{3}\right\}\right)\right. \nonumber \\
&&+ \mathcal{DS}\left(\left\{\ket{v_+}, \ket{v^{\perp}_{+}}\right\},\left\{\frac13 ,\frac{1}{3}\right\}\right) \nonumber\\
&&\left. + \mathcal{DS}\left(\left\{\ket{v_-}, \ket{v^\perp_-}\right\},\left\{\frac13,\frac{1}{3}\right\}\right)\right]. \nonumber 
  \eea
Leveraging \eqref{DSpsi12}, this reduces to,
\bea
& \mathcal{DME}= \frac12\left[\frac13+\frac13+\frac13 +\frac13+\frac13+\frac13\right]
 = 1.
\eea
Henceforth, theorem \ref{trineproof} is proved.
\end{proof}
One can be curious about the highest gap possible between the values of $\mathcal{DME}$ and $\mathcal{D\overline{M}S}$ with these two classes of measurements. For this, we will choose the unitaries such that $U_1\ket{0}=\ket{1}, U_2\ket{v_+}=\ket{v_+^{\perp}}$ and $U_3\ket{v_-}=\ket{v_-^{\perp}}$. With these choices of unitaries, $\mathcal{D\overline{M}S}=\frac12$ because it eventually reduces to the distinguishability of two same states at each outcome. From \eqref{dmelatt2}, we can see $\mathcal{DME}$ does not depend on $U_i$. So $\mathcal{DME}$ is still $1$. Now we will show an example the other way around, i.e., $\mathcal{DME} < \mathcal{D\overline{M}S} =1$.
\begin{thm}\label{th12}
    For the two measurements defined by '\textit{Right Asymmetric Trine}'\eqref{RAT} and '\textit{Reverse Trine}'\eqref{trine2}, $\mathcal{DME} < \mathcal{D\overline{M}S} =1$.
\end{thm}
\begin{proof}
   The expression of $\mathcal{DME}$ between \eqref{RAT} and \eqref{trine2} with the most general entangled state as the initial state is given by \eqref{dme<1}.
   For $\mathcal{DME}$ to be $1$, distinguishability of all three pairs of reduced states of \eqref{dme<1} should be the sum of their respective probabilities, i.e., each pair should consist of orthogonal states. Applying this condition, we find from the first pair,
   \bea\label{s1}
   \left(\left|c\right|^2+\left|d\right|^2\right)\sin\frac{\theta}{2}+z\cos\frac{\theta}{2}=0.
   \eea
   From the second pair,
   \bea\label{s2}
   \frac{\sqrt{3}}{4}\left(|a|^2+|b|^2-|c|^2-|d|^2\right)+\frac34\Bar{z}-\frac14 z=0.
   \eea
   From the third pair,
   \bea\label{s3}
   \frac{\sqrt{3}}{4}\left(|a|^2+|b|^2-|c|^2-|d|^2\right)-\frac34\bar{z}+\frac14 z=0,
   \eea
   where $z=ac^*+bd^*$.
   If \eqref{s2} and \eqref{s3} has to be true simultaneously, then we get from adding \eqref{s2} and \eqref{s3},
   \bea\label{s4}
   &&|a|^2+|b|^2=|c|^2+|d|^2,
   \eea
     subtracting \eqref{s3} from \eqref{s2},
\bea\label{s5}
   3\bar{z}=z.
   \eea
   \eqref{s5} will be true if and only if $z=0$. Putting this into \eqref{s1}, we arrive at,
   \be
   \left(|c|^2+|d|^2\right)\sin\frac{\theta}{2}=0.
   \ee
   From the construction of our measurement, we know $\theta \neq 0$. So, $(|c|^2+|d|^2)$ needs to be zero which is not possible as from \eqref{s4}, $|a|^2+|b|^2=|c|^2+|d|^2=0$ and that is contradicting the normalization condition. So, \eqref{s1},\eqref{s2} and \eqref{s3} can not be satisfied simultaneously. Therefore, from \eqref{DME},
   \be
   \mathcal{DME}<\frac12(p_1+p_2+p_3+p_4+p_5+p_6)=1.
   \ee
  \\
    With these two measurements of theorem \ref{th12}, from \eqref{DIS}, we can derive,
    \bea
     &&\mathcal{D\overline{M}S}\nonumber\\
     &=& \frac12\max_{\ket{\zeta}} \left(\mathcal{DS}\left[\Big\{\ket{0},\ket{1}\Big\},\left\{\alpha|\langle \zeta|U|0\rangle|^2,\frac{2}{3}|\langle \zeta|1\rangle|^2\right\}\right] \right. \nonumber\\
    &&+ \mathcal{DS}\left[\left\{\ket{v_+},\ket{v_+^{\perp}}\right\},\left\{\beta|\langle \zeta|v_+\rangle|^2,\frac{2}{3}|\langle \zeta|v_+^{\perp}\rangle|^2\right\}\right]\nonumber\\
    &&\left. + \mathcal{DS}\left[\left\{\ket{v_-},\ket{v_-^{\perp}}\right\},\left\{\gamma|\langle \zeta|v_-\rangle|^2,\frac{2}{3}|\langle \zeta|v_-^{\perp}\rangle|^2\right\}\right]\right).\nonumber\\
    \eea
    So, regardless of the choice of the initial state $\ket{\zeta}$, it satisfies $(i)$ of theorem \eqref{th3} and we can infer,
    $\mathcal{D\overline{M}S}=1.$\\  
\end{proof}
 To provide explicit values of $\mathcal{DME}$, we take $\theta=\frac{13\pi}{45}$ from the numerical analysis presented in Fig.\ref{DMEplot}. where $\mathcal{DME}=0.995$ and $\mathcal{D\overline{M}S}=1$. 

From \eqref{all}, we came to know that $\mathcal{D\overline{M}E}$ can be more advantageous than both of $\mathcal{DME}$ and $\mathcal{D\overline{M}S}$. We provide an example in this direction. 
\begin{thm}\label{th13}
     For the measurements defined by '\textit{Left-right Asymmetric Trine}'\eqref{LRAT} and '\textit{Reverse Trine}'\eqref{trine2},
     $\mathcal{DME} < 1, \mathcal{D\overline{M}S} < 1$, but $\mathcal{D\overline{M}E}=1$ if and only if all the unitaries $\{U_i\}, i=1,2,3$ are not $\I$. 
\end{thm}
\begin{proof}
    The expression of $\mathcal{DME}$ is exactly like that of \eqref{dme<1}, and we proved $\mathcal{DME}<1$ in the theorem \ref{th12}.\\
    With these two measurements, \eqref{DIS} can be simplified like the following:
    \begin{widetext}
    \bea
     \mathcal{D\overline{M}S} 
    &=& \frac12 \left(\mathcal{DS}\left[\Big\{\ket{0},\ket{1}\Big\},\left\{\alpha|\langle \zeta|U|0\rangle|^2,\frac{2}{3}|\langle \zeta|1\rangle|^2\right\}\right]  
    + \mathcal{DS}\left[\left\{U_2\ket{v_+},\ket{v_+^{\perp}}\right\},\left\{\beta|\langle \zeta|v_+\rangle|^2,\frac{2}{3}|\langle \zeta|v_+^{\perp}\rangle|^2\right\}\right] \right. \nonumber\\
    &&\left. + \mathcal{DS}\left[\left\{U_3\ket{v_-},\ket{v_-^{\perp}}\right\},\left\{\gamma|\langle \zeta|v_-\rangle|^2,\frac{2}{3}|\langle \zeta|v_-^{\perp}\rangle|^2\right\}\right]\right).\nonumber\\
    \eea
    \end{widetext}
    As $U_2,U_3 \neq \I$, there are two pair of non orthogonal post-measurement states with at least one pair has two non-zero probabilities, irrespective of the choice of $\zeta$. That contradicts $(i)$ of theorem \ref{th2}. So, $\mathcal{D\overline{M}S}<1$.

    For $\mathcal{D\overline{M}E}$, we put these two measurements into \eqref{DIE} and we take the best initial state as $\ket{\phi^+}=\frac{1}{\sqrt{2}}(\ket{00}+\ket{11})$. Now, \eqref{equaldbarme} looks like,

\bea\label{th13c}
\mathcal{D\overline{M}E}&=& \frac12 \left(\mathcal{DS}\left[\Big\{\ket{0}U\ket{0},\ket{1}\ket{1}\Big\},\left\{\frac{\alpha}{2},\frac13\right\}\right] \right. \nonumber\\
 &&+ \mathcal{DS}\left[\left\{U_2\ket{v_+}\ket{v_+},\ket{v_+^{\perp}}\ket{v_+^{\perp}}\right\},\left\{\frac{\beta}{2},\frac13\right\}\right]\nonumber\\
&&\left. + \mathcal{DS}\left[\left\{U_3\ket{v_-}\ket{v_-},\ket{v_-^{\perp}}\ket{v_-^{\perp}}\right\},
\left\{\frac{\gamma}{2},\frac13 \right\}\right]\right).\nonumber\\
\eea
 \   
    As it reduced to three pairs of orthogonal states, by the grace of \eqref{DSpsi12}, \eqref{th13c} leads to,
    \bea
    & \mathcal{D\overline{M}E}=\frac12\left[\frac{\alpha}{2}+\frac13+\frac{\beta}{2}+
     \frac13+\frac{\gamma}{2} +\frac13\right]=1,
    \eea
which completes the proof.    
\end{proof}
We see the strict implication of \eqref{all} for the distinguishability of measurements and now we move into antidistinguishability of measurements. \\

In section \ref{SEC III}, we have shown that for a pair of binary outcome measurements, $\mathcal{AMS}$ is identical with $\mathcal{DMS}$. The other three scenarios eventually involve the antidistinguishability of the states. So, for antidistinguishability, we need to have at least three measurements to make these arrangements different from distinguishability.\\
From section \ref{SEC III}, it also can be easily commented that $\mathcal{AMS}$ for qubit projective measurements is always less than $1$. We will show availing the entangled system or post-measurement state, antidistinguishability will be better than that in the case of $\mathcal{AMS}$.

\begin{thm}\label{th14}
    For the three qubit projective measurements defined as $F_{1|1}=\ket{0}\bra{0}, F_{1|2}=\ket{v_+}\bra{v_+}, F_{1|3}=\ket{v_-}\bra{v_-}$,
      $\mathcal{AMS} <1$, $\mathcal{A\overline{M}S} = \mathcal{AME} = 1$.
\end{thm}
\begin{proof}
     If we put these three measurements into \eqref{amsprojective}, it takes the form (with equiprobable distribution),
    \bea
     \mathcal{AMS}&=& 1-\frac13\min\left\{\left(1-\sqrt{1-|\langle 0|v_+\rangle|^2}\right),\right. \nonumber\\
    &&\left. \left(1-\sqrt{1-|\langle 0|v_-\rangle|^2}\right),
    \left(1-\sqrt{1-|\langle v_-|v_+\rangle|^2}\right) \right\}\nonumber\\
    &=&\frac23+\frac{1}{2\sqrt{3}}\nonumber\\
    &<& 1.
    \eea
     With these three measurements and $\ket{\zeta}$ as the initial state, \eqref{AbarMS} becomes,
\bea\label{thm14part}
 \mathcal{A\overline{M}S} &=& \frac13\max_{\ket{\zeta}}\left(\mathcal{AS}\left[\Big\{\ket{0},\ket{v_+},\ket{v_-}\Big\},\right.\right. \nonumber\\
&&\left. 
\left\{|\langle \zeta|0\rangle|^2,|\langle \zeta|v_+\rangle|^2,|\langle \zeta|v_-\rangle|^2\right\}\right]\nonumber\\
&& +\mathcal{AS}\left[\left\{\ket{1},\ket{v_+^{\perp}},\ket{v_-^{\perp}}\right\},\right. \nonumber\\
&& \left. \left. \left\{|\langle \zeta|1\rangle|^2,|\langle \zeta|v_+^{\perp}\rangle|^2,|\langle \zeta|v_-^{\perp}\rangle|^2\right\}\right]\right).
\eea
As $\{\ket{0},\ket{v_+},\ket{v_-}\}$ and $\{\ket{1},\ket{v_+^{\perp}},\ket{v_-^{\perp}}\}$ are two sets of antidistinguishable states, $(i)$ of theorem \ref{th3} tells that $\mathcal{A\overline{M}S}=1$.

For $\mathcal{AME},$ the bell state $\ket{\phi^+}$ is taken as the best state. Consequently, \eqref{AME} becomes, 
\bea\label{ameth14}
    \mathcal{AME} &=& \frac13\left(\mathcal{AS}\left[\Big\{\ket{0},\ket{v_+},\ket{v_-}\Big\},\left\{\frac12,\frac12,\frac12\right\}\right] \right. \nonumber\\
   &&\left. + \mathcal{AS}\left[\left\{\ket{1},\ket{v_+^{\perp}},\ket{v_-^{\perp}}\right\},
    \left\{\frac12,\frac12,\frac12\right\}\right]\right).\nonumber\\
\eea
This is the sum of two triplets of antidistinguishable states. With the help of \eqref{ADSpsi123}, \eqref{ameth14} reduces to,
\bea
\mathcal{AME}=\frac13\left(6\times\frac12\right)=1.
\eea
Thus, theorem \ref{th14} is demonstrated.
\end{proof}

\begin{thm}\label{thm20}
     For the three qubit projective measurements defined as $F_{1|1}=\ket{0}\bra{0}, F_{1|2}=\ket{+}\bra{+}, F_{1|3}=\ket{\nu}\bra{\nu},(\ket{\nu} =\cos\left(\frac{\theta}{2}\right)\ket{0}+\sin\left(\frac{\theta}{2}\right)\ket{1})$,
     $\mathcal{AME} < 1, \mathcal{A\overline{M}S} < 1$ but $\mathcal{A\overline{M}E}=1$ if $\theta \in (0,\pi)\cup(3\pi/2,2\pi)$ and $\kappa =(\sin\theta+\cos\theta)$ obeys these two inequalities following:
     \bea\label{zetacond}
     &(i)& \kappa < 0,\nonumber\\
     &(ii)& \kappa^2 \geqslant \left(\frac{\kappa+1}{2}\right)^4,\nonumber\\
     \eea
\end{thm}
\begin{proof}
     With these three measurements and $\{p_x\}_x=\frac13$, \eqref{AbarMS} simplified into,
\bea\label{abarmsproj}
 \mathcal{A\overline{M}S} &=& \frac13\max_{\ket{\zeta}}\left(\mathcal{AS}\left[\Big\{\ket{0},\ket{+},\ket{\nu}\Big\},\right. \right.\nonumber\\
&& \left. \left\{|\langle \zeta|0\rangle|^2,|\langle \zeta|+\rangle|^2,|\langle \zeta|\nu\rangle|^2\right\}\right]\nonumber\\
&&+\mathcal{AS}\left[\left\{\ket{1},\ket{-},\ket{\nu^{\perp}}\right\},\right.\nonumber\\
&&\left. \left. \left\{|\langle \zeta|1\rangle|^2,|\langle \zeta|-\rangle|^2,|\langle \zeta|\nu^{\perp}\rangle|^2\right\}\right]\right),
\eea
which will be less than 1 if at least any of these triplets $\{\ket{0},\ket{+},\ket{\phi}\}$ or $\{\ket{1},\ket{-},\ket{\phi^\perp}\}$ is not antidistinguishable. Without loss of generality, we take the first triplet to be not antidistinguishable. From that, we can infer $\ket{\phi}$ should be on the X-Z plane of the Bloch sphere and $\theta \in (0,\pi)\cup(3\pi/2,2\pi)$, which is consistent with the range of $\theta$ mentioned in the theorem. So $\mathcal{A\overline{M}S}< 1.$\\

Now, we want to take a look at $\mathcal{AME}$ with these three measurements. Let us take the general entangled state $\ket{\phi}=a\ket{00}+b\ket{01}+c\ket{10}+d\ket{11}$ as the best possible state where $|a|^2+|b|^2+|c|^2+|d|^2 =1$ by normalization. In this case, \eqref{AME} can be put down as, \\
\begin{widetext}
     \bea\label{ameproj}
 \mathcal{AME} &=& \frac13\Bigg(\mathcal{AS}\Bigg[\left\{\frac{a\ket{0}+b\ket{1}}{\sqrt{p_1}},\frac{\frac{1}{\sqrt{2}}(a+c)\ket{0}+\frac{1}{\sqrt{2}}(b+d)\ket{1}}{\sqrt{p_2}},
 \frac{(a\cos{\frac{\theta}{2}}+c\sin{\frac{\theta}{2}})\ket{0}+(b\cos{\frac{\theta}{2}}+d\sin{\frac{\theta}{2}})\ket{1}}{\sqrt{p_3}}\right\},\nonumber\\
&& \left\{p_1,p_2,p_3\right\}\Bigg]+\mathcal{AS}\Bigg[\left\{\frac{c\ket{0}+d\ket{1}}{\sqrt{q_1}},\frac{\frac{1}{\sqrt{2}}(a-c)\ket{0}+\frac{1}{\sqrt{2}}(b-d)\ket{1}}{\sqrt{q_2}},
 \frac{(a\sin{\frac{\theta}{2}}-c\cos{\frac{\theta}{2}})\ket{0}+(b\sin{\frac{\theta}{2}}-d\cos{\frac{\theta}{2}})\ket{1}}{\sqrt{q_3}}\right\}\nonumber\\
&& \left\{q_1,q_2,q_3\right\}\Bigg]\Bigg).
\eea
\end{widetext}
The probabilities are follows:
 \bea
 &&p_1=a^2+b^2, p_2=\frac12\left\{(a+c)^2+(b+d)^2\right\}, \nonumber\\
 &&p_3=\left(a\cos{\frac{\theta}{2}}+c\sin{\frac{\theta}{2}}\right)^2+\left(b\cos{\frac{\theta}{2}}+d\sin{\frac{\theta}{2}}\right)^2,\nonumber\\
 &&q_1=c^2+d^2, q_2=\frac12\left\{(a-c)^2+(b-d)^2\right\}, \nonumber\\
 &&q_3=\left(a\cos{\frac{\theta}{2}}-c\sin{\frac{\theta}{2}}\right)^2+\left(b\cos{\frac{\theta}{2}}-d\sin{\frac{\theta}{2}}\right)^2.\nonumber
\eea

By iterating all the values of $(a,b,c,d)$ under the normalization constraint, we check the conditions of \eqref{condAS} for both outcomes. We have done this numerical analysis by fixing the value of $\theta \in (0,\pi)\cup (3\pi/2,2\pi)$ in small intervals. We find that there is no such set of $(a,b,c,d)$ for any $\theta$ in this range, for which both the triplets will be antidistinguishable. So we conclude, for these three measurements, $\mathcal{AME} < 1$ in the given range of $\theta$. \\

For $\mathcal{A\overline{M}E}$, we take Bell state $\ket{\phi^+}$ as the best possible state. From \eqref{AbarME}, we can write with $\{p_x\}_x=\frac13$ ,
\bea\label{AbarME=1}
&& \mathcal{A\overline{M}E}\nonumber\\
&=& \frac13\max_{\ket{\phi^+}}\left(\mathcal{AS}\left[\left\{\ket{0}\ket{0},\ket{+}\ket{+},\ket{\phi}\ket{\phi}\right\},
 \left\{\frac12,\frac12,\frac12\right\}\right]\right. \nonumber\\
 &&\left. + \mathcal{AS}\left[\left\{\ket{1}\ket{1}, \ket{-}\ket{-}, \ket{\phi^{\perp}}\ket{\phi^{\perp}}\right\},  \left\{\frac12,\frac12,\frac12\right\}\right]\right).\nonumber\\
\eea
 Both the triplets should satisfy the conditions of \eqref{condAS}, so that $\mathcal{A\overline{M}E} = \frac13\big(6\times\frac12 \big)=1.$ For the first triplet to be antidistinguishable, the sufficient conditions (\eqref{condAS}) appear as,
 
 \bea\label{suba} 
&& x_1^2+x_2^2+x_3^2<1,\nonumber\\
&& \left(x_1^2+x_2^2+x_3^2-1\right)^2 \geqslant 4x_1^2x_2^2x_3^2,\nonumber\\
 \eea
  where, $x_1=\frac12, x_2=\cos^2\frac{\theta}{2}, x_3=\cos^2(\frac{\theta}{2}-\frac{\pi}{4})$.\\

 If we impose the conditions of antidistinguishability of the second triplet, we will get the exactly same inequalities like \eqref{suba}.  By simplification, \eqref{suba} becomes \eqref{zetacond} which is the sufficient condition with the given range of $\theta$ for theorem \ref{thm20} to be true.
\end{proof}

We could not confirm any vantage point between $\mathcal{AME}$ and $\mathcal{A\overline{M}S}$ from theorems \ref{th14} and \ref{thm20}. So we move into POVMs if there are some examples where we can show certain advantages between these two scenarios. For this exploration, we propose some new measurements, which are relabelled versions of the previous measurements, which are defined at the beginning of this section. From \eqref{RAT}, we take one combination of outcomes and denote those measurements by $\{M_a(\theta)\}$:
\bea\label{RAT1}
& M_1=\beta^{\frac{1}{2}}\ket{v_+}\!\bra{v_+}, M_2=\gamma^{\frac{1}{2}}\ket{v_-}\!\bra{v_-},\nonumber\\
& M_3=\alpha^{\frac{1}{2}}\ket{0}\!\bra{0}U.
\eea
Similarly, another combination of \eqref{LRAT} is taken with different unitary with respect to \eqref{RAT1} and denoted by $\{N_a(\Theta)\}$:
\bea\label{LRAT1}
 & N_1=\gamma^{\frac{1}{2}}U_3\ket{v_-}\!\bra{v_-}, N_2=\alpha^{\frac{1}{2}}\ket{0}\!\bra{0}U',\nonumber\\
 & N_3=\beta^{\frac{1}{2}}U_2\ket{v_+}\!\bra{v_+}.
\eea
and $U'\ket{0}=\ket{\phi(\Theta)}$.

\begin{thm}\label{th17}
For the measurements defined by '\textit{Trine}'\eqref{trine1},
\eqref{RAT1} and \eqref{LRAT1}, $\mathcal{A\overline{M}S}<  \mathcal{AME}=1$ with a certain choice of unitaries such that $U_2\ket{v_+}=\ket{0}, U_3\ket{v_-}=\ket{0}$.\\

\end{thm}
\begin{proof}
      With initial state $\ket{\zeta}$, \eqref{AbarMS} expands as,
\bea\label{thm17part}
&&\mathcal{A\overline{M}S} \nonumber\\
&=& \frac{1}{3} \max_{\ket{\zeta}} \left(\mathcal{AS}\left[\Big\{\ket{0},\ket{v_+},U_3\ket{v_-}\Big\},\right.\right.\nonumber\\
&&\left.\left\{\frac{2}{3}|\langle \zeta|0\rangle|^2,\beta(\theta)|\langle \zeta|v_+\rangle|^2,\gamma(\Theta)|\langle \zeta|v_-\rangle|^2\right\}\right]\nonumber\\
&&+\mathcal{AS}\left[\Big\{\ket{v_+},\ket{v_-},\ket{0}\Big\},\right.\nonumber\\
&&\left.\left\{\frac23|\langle \zeta|v_+\rangle|^2,\gamma(\theta)|\langle \zeta|v_-\rangle|^2,\alpha(\Theta)|\langle \zeta|U'|0\rangle|^2\right\}\right]\nonumber\\
&&+\mathcal{AS}\left[\Big\{\ket{v_-},\ket{0},U_2\ket{v_+}\Big\},\right.\nonumber\\
&&\left.\left\{\frac23|\langle \zeta|v_-\rangle|^2,\alpha(\theta)|\langle \zeta|U|0\rangle|^2,\beta(\Theta)|\langle \zeta|v_+\rangle|^2\right\}\right].\nonumber\\
\eea

With $U_2\ket{v_+}=\ket{0}$ and $U_3\ket{v_-}=\ket{0}$, first and third triplet becomes not antidistinguishable. Irrespective of the chosen initial state, it does not satisfy theorem \ref{th3}. So, $\mathcal{A\overline{M}S} < 1$.

For $\mathcal{AME},$ we take the Bell state $\ket{\phi^+}$ as the best initial state. With the same set of measurements, \eqref{AME} can be brought into,
\bea
    \mathcal{AME} &=& \frac13\left(\mathcal{AS}\left[\left\{\ket{0},\ket{v_+},\ket{v_-}\right\},\left\{\frac13,\frac{\beta(\theta)}{2},\frac{\gamma(\Theta)}{2}\right\}\right]\right. \nonumber\\
   &&+\mathcal{AS}\left[\left\{\ket{v_+},\ket{v_-},\ket{\phi(\Theta)}\right\},
    \left\{\frac13,\frac{\gamma(\theta)}{2},\frac{\alpha(\Theta)}{2}\right\}\right]\nonumber\\
  &&\left. + \mathcal{AS}\left[\left\{\ket{v_-},\ket{\phi(\theta)},\ket{v_+}\right\},
    \left\{\frac13,\frac{\alpha(\theta)}{2},\frac{\beta(\Theta)}{2}\right\}\right]\right).\nonumber\\
   \eea
   From \eqref{condAS}, it is easy to check $\{\ket{0},\ket{v_+},\ket{v_-}\}$, $\{\ket{v_+},\ket{v_-},\ket{\phi(\Theta)}\}$ and $\{\ket{v_-},\ket{\phi(\theta)},\ket{v_+}\}$ are antidistinguishable. Subsequently,
   \bea 
   \mathcal{AME}&=&\frac13\left(3\times\frac13+ \frac12\Big\{\alpha(\theta)+\beta(\theta)+ \gamma(\theta) \right.  \nonumber\\
   &&\left. + \alpha(\Theta) +\beta(\Theta) + \gamma(\Theta) \Big\} \right) =1.
   \eea
   Therefore, theorem \ref{th17} is proved.
\end{proof}

Similarly, like the analogous theorem of distinguishability, we are here interested in the gap between the values of $\mathcal{A\overline{M}S}$ and $\mathcal{AME}$. antidistinguishability does not have any closed form in this case. So, we make use of semi-definite programming. First, we fix three measurements, that is, fixing the values of $\theta$ and $\Theta$. For these three measurements with $\theta=\frac{\pi}{6}, \Theta=\frac{\pi}{12}$, from SDP, we found $\mathcal{A\overline{M}S}=0.923$. The best initial state happens to be $\ket{\zeta}=-0.9991\ket{0}+0.0416\ket{1}$.  $\mathcal{AME}$ does not depend on unitaries. So, irrespective of the values of $\theta,\Theta$, $\mathcal{AME}=1$. 

As usual, we want an example of the reverse statement of the theorem \ref{th17}. For this thing, some measurements are needed to be relabelled. From \eqref{RAT}, we take one combination of outcomes and denote those measurements by $\{R_a(\theta)\}$:
\bea\label{RAT2}
& R_1=\gamma^{\frac{1}{2}}\ket{v_-}\!\bra{v_-}, R_2=\beta^{\frac{1}{2}}\ket{v_+}\!\bra{v_+},\nonumber\\
& M_3=\alpha^{\frac{1}{2}}\ket{0}\!\bra{0}U.
\eea
From \eqref{LRArT}, we take other combinations of outcomes. This is denoted by $\{S_a(\mu)\}$:
\bea\label{LRArT2}
& S_1 = c^\frac12 V_3\ket{v_-^{\perp}}\!\bra{v_-^{\perp}}, S_2 = b^\frac12\ket{v_+^{\perp}}\!\bra{v_+^{\perp}}V_2,\nonumber\\
& S_3= a^\frac{1}{2}V_1\ket{1}\!\bra{1}.
\eea

\begin{thm}\label{th18}
For the measurements defined by '\textit{Trine}'\eqref{trine1}, \eqref{RAT2} and \eqref{LRArT2}, $\mathcal{A\overline{M}S}= 1$ but $\mathcal{AME}<1$, if $V_1\ket{1}=\ket{v_+}$ and $V_3\ket{v_-^{\perp}}=\ket{v_+}$.
\end{thm} 
\begin{proof}
    With these three measurements, from \eqref{AbarMS}, we can write,
\bea
 \mathcal{A\overline{M}S} &=& \frac{1}{3}\max_{\ket{\zeta}}\left(\mathcal{AS}\left[\left\{\ket{0},\ket{v_-},V_3\ket{v_-^\perp}\right\},\right. \right.\nonumber\\
 &&\left. \left\{\frac{2}{3}|\langle \zeta|0\rangle|^2,\gamma|\langle \zeta|v_-\rangle|^2,c|\langle \zeta|v_-^{\perp}\rangle|^2\right\}\right]\nonumber\\
 &&+\mathcal{AS}\left[\left\{\ket{v_+},\ket{v_+},\ket{v_+^\perp}\right\},\right. \nonumber\\
 &&\left. \left\{\frac{2}{3}|\langle \zeta|v_+\rangle|^2,\beta|\langle \zeta|v_+\rangle|^2,b|\langle \zeta|V_2|v_+^\perp\rangle|^2\right\}\right]\nonumber\\
 &&+ \mathcal{AS}\left[\Big\{\ket{v_-},\ket{0},V_1\ket{1}\Big\},\right.\nonumber\\
&&\left. \left. \left\{\frac{2}{3}|\langle \zeta|v_-\rangle|^2,\alpha|\langle \zeta|U|0\rangle|^2,a|\langle \zeta|1\rangle|^2\right\}\right]\right).\nonumber\\
\eea
The choice of $V_1$ and $V_3$ of theorem \ref{th18} makes two triplets, $\{\ket{0},\ket{v_-},V_3\ket{v_-^{\perp}}\}$ and $\{\ket{v_-},\ket{0},V_1\ket{1}\}$ antidistinguishable.  We choose the best possible state as $\ket{\zeta}=\ket{v_+^{\perp}}$, so two probabilities in the second triplet becomes zero and as a result, there will be only one state $\ket{v_+^{\perp}}$. So, the second triplet is also antidistinguishable. So, all the three triplets are antidistinguishable, so, from theorem \ref{th3}, we can say $\mathcal{A\overline{M}S}=1.$ \\

For $\mathcal{AME},$ Alice and Bob can take any general entangled state as the best initial state. For $\mathcal{AME}$ to be $1$, the reduced state for each outcome at Bob's side should be antidistinguishable. For the second outcome of first two measurements, the same operator is acting on Alice's part of entangled state. So, irrespective of the initial entangled state, Bob's part will be reduced to same state for this case. Now, Bob has one triplet (2nd) where two states are the same. So, this triplet is not antidistinguishable. As they confirm that one out of three triplet is always not antidistinguishable, we can conclude $\mathcal{AME}<1$. This completes the proof.

\end{proof}
Like the previous theorem, we want to cite a specific example for this instance. We take $\mu=\frac{23\pi}{12}$ and $\theta=-\frac{\pi}{6}$ and find $\mathcal{AME}=0.9954$ by SDP. The initial entangled state is found to be $\ket{\zeta}=\sqrt{0.88}\ket{\eta}\ket{0}+\sqrt{0.12}\ket{\eta^\perp}\ket{1}$, where $\ket{\eta}=0.17\ket{0}+0.9854\ket{1}$. We do not hamper $V_1,V_3$, so, $\mathcal{A\overline{M}S}$ is still $1$. 
 

\section{Conclusion}
In this paper, we developed a general framework for distinguishability and antidistinguishability of quantum measurements. To distinguish or antidistinguish a known set of measurements, we need the best possible initial state (single or entangled) to be fed into the measurement device. Depending on the setup, we may or may not have access to the post-measurement state after one of the measurements from the set has been carried out. So, there are four scenarios for distinguishing or antidistinguishing quantum measurements, considering all possible combinations of the initial state and the availability of the post-measurement state. For any pair of qubit projective measurements, access to the post-measurement state with an initial single state provides an advantage over an entangled initial state without the post-measurement state. There exist some qubit non-projective measurements for which maximum distinguishability for entangled systems without access to post-measurement state is obtained using non-maximally entangled states. For any set of measurements, using an initial entangled state and access to the post-measurement state provides the best distinguishing or antidistinguishing probability, with the least being starting with single systems and without the post-measurement state. There is no strict relationship between the scenario of probing with entangled systems but without access to the post-measurement state and the scenario of probing with single systems and with access to the post-measurement state. There are measurements that are perfectly distinguishable or antidistinguishable in one scenario but not in other(s). We have introduced different variants of a well-known \textit{'trine'} qubit POVM to construct examples such that they are perfectly distinguishable or antidistinguishable using an entangled state without the post-measurement state, but not with single systems with an available post-measurement state, and vice-versa. There exist qubit measurements that can be perfectly distinguished (or antidistinguished) only using an initial entangled system with access to the post-measurement state but not in other frameworks.

Let us point out some possible directions in future research. In the course of the proof of theorem \ref{DbarMSgeqDME}, we numerically show that the maximally entangled state is the best initial state for $\mathcal{DME}$ of a pair of qubit projective measurements. But, if we want to analyse the projective measurements of higher dimensions, the numerical method would be cumbersome. So, it would be an interesting task to find analytical proof for the sufficiency of maximally entangled states for any number of qudit projective measurements. The same task can be done for $\mathcal{AME}$. Interestingly, all our results in section \ref{SEC V} can be shown using qubit measurements. One can try to find examples of higher dimensional measurements that are perfectly distinguishable or antidistinguishable in one scenario but not in other(s). Thus, there will be a possibility to find the families of measurements in different dimensions to fulfill the earlier condition. Alongside, one can quantify the advantage of one scenario over the other for both the distinguishability and antidistinguishability across the dimension and then find the class of measurements for which the ratios between $\mathcal{D\overline{M}S}$ and $\mathcal{DMS}$, $\mathcal{DME}$ and $\mathcal{DMS}$, $\mathcal{D\overline{M}E}$ and $\mathcal{DME}$ are unbounded with respect to dimension. It will also be interesting to consider the multiple-shot regime for studying distinguishability and antidistinguishability in scenarios with or without access to the post-measurement state. Furthermore, one can think of the application of these quantities in information-theoretic tasks. 

%

\subsection*{Acknowledgements}
This work is supported by STARS (STARS/STARS-2/2023-0809), Govt.
of India.

\bibliography{ref}

\end{document}